\documentclass[journal]{IEEEtran}

\IEEEoverridecommandlockouts                              

\usepackage{amsmath, amsfonts, amssymb}
\usepackage{color}
\usepackage{graphicx,tikz}
\usetikzlibrary{arrows.meta}
\tikzset{>={Latex[width=2mm,length=2mm]}}

\newcommand{\E}{\mathsf{E}}
\newcommand{\R}{\mathbb{R}}
\newcommand{\N}{\mathcal{N}}
\newcommand{\eg}{{i.e.}, }
\newcommand{\I}{\mathfrak{I}}
\newcommand{\U}{\mathcal{U}}
\newcommand{\Q}{\mathcal{Q}}
\newcommand{\bbN}{\mathbb{N}}
\newcommand{\p}{\mathcal{P}}
\newcommand{\Y}{\mathcal{Y}}
\newcommand{\X}{\mathcal{X}}
\newcommand{\s}{\mathcal{S}}
\newcommand{\A}{\mathcal{A}}
\newcommand{\W}{\mathcal{W}}

\newcommand{\T}{^{\mbox{\tiny \sf T}}}
\newcommand{\tr}{{\rm{tr}}}

\newtheorem{thm}{Theorem}[section]
\newtheorem{lm}[thm]{ Lemma}

\newtheorem{pr}[thm]{ Proposition}
\newtheorem{rem}[thm]{ Remark}

\newtheorem{corr}[thm]{ Corollary}

\title{\LARGE \bf 
Optimal Controller Synthesis and  Dynamic Quantizer Switching\\ for 
Linear-Quadratic-Gaussian Systems
}

\author{Dipankar Maity and Panagiotis Tsiotras
\thanks{}
\thanks{The authors are with the Guggenheim School of Aerospace Engineering, Georgia Institute of Technology, Atlanta, GA, 30332, USA.
        {Email: \small dmaity@gatech.edu, tsiotras@gatech.edu}}%
\thanks{This work has been supported by ARL under DCIST CRA W911NF-17-2-0181 and by ONR award N00014-18-1-2375.}
}

\begin{document}

\maketitle
\thispagestyle{empty}
\pagestyle{empty}

\begin{abstract}
In networked control systems, often the sensory signals are quantized before being transmitted to the controller. Consequently,  performance is affected by the coarseness of this quantization process. Modern communication technologies allow users to obtain resolution-varying quantized measurements based on the prices paid. 
 In this paper, we consider optimal controller synthesis of a Quantized-Feedback Linear-Quadratic-Gaussian (QF-LQG) system where the measurements are to be quantized before being transmitted to the controller.
  The system is presented with several choices of quantizers, along with the cost of operating each quantizer. 
  The objective is to jointly select the quantizers and the controller that would maintain an optimal balance between control performance and quantization cost. 
  Under certain assumptions,  this problem can be decoupled into two optimization problems: one for optimal controller synthesis and the other for optimal quantizer selection.
   We show that, similarly to the classical LQG problem, the optimal controller synthesis subproblem is characterized by  Riccati equations.
    On the other hand, the optimal quantizer selection policy is found by solving a certain Markov-Decision-Process (MDP). 
\end{abstract}

\section{INTRODUCTION}

Increasingly, many control systems nowadays consist of multiple sensors, actuators and plants that are spatially distributed. Control of such  systems requires uninhibited and reliable exchange of signals among these components over a shared communication network. 
Often, the underlying communication network suffers from several limitations such as insufficient bandwidth, noisy transmissions, or delays.
Although many of these abovementioned limitations can be alleviated with current advancements in communication technologies, at the same time  it may be expensive to deploy such a communication infrastructure. 
Therefore, the performance of such systems no longer depends solely on the controller structure but also on the underlying communication infrastructure and the associated communication cost. 

In a typical communication framework, signals are quantized (encoded) before being transmitted through a channel. Upon receiving the transmitted signal, a reconstruction (decoding) is performed to estimate the original signal. Thus, the quality of the quantization dictates the distortion in the reconstructed signal. 
Higher resolution  quantization results in lower distortion. This typically requires a higher number of bits to represent the quantized signal, and hence higher channel bandwidth for transmission. Depending on the criticality of the task, at certain times high-resolution quantization may be required, while other times coarser resolution is sufficient. Therefore, the quantization selection must be adapted optimally over the time horizon to meet the expected quantization-resolution of the transmitted signal with minimal use of the communication resources. 

In this work we consider the classical LQG control problem under quantization constraints which can be traced back to \cite{ borkar1997lqg, tatikonda2000control, tatikonda1998control, nair2004stabilizability, tatikonda2004stochastic}. 
While in most of the prior works, the emphasis has been on the joint design of controller and quantizer, in this paper we take a different approach where instead of designing a quantizer, we formulate a quantizer scheduling problem. 
In \cite{nair2004stabilizability} and other related works, the necessity for a time-varying quantizer for stability of linear systems has been studied. 
Recent works, such as \cite{kostina2019rate}, also show the tradeoff between control cost and communication data-rate in the context of infinite-horizon LQG problems.
Although LQG optimal control with quantized measurements has been studied for decades, the optimal structure of the controller and quantizer is, however, still unknown. 
Approximate solutions to the optimal quantizer and controller (for infinite time horizon) synthesis problem have been constructed under restrictive assumptions on the quantization schemes such as lattice quantization \cite{kostina2019rate}, entropy coded dithered quantizer for single-input-single-output systems \cite{silva2010framework,silva2015characterization} etc.

One may alternatively think of the optimal quantizer design problem as an equivalent problem of \textit{selecting} the optimal quantizer in the space of quantizers. 
In this work, instead of designing the quantizers (or equivalently finding the optimal quantizer in the space of quantizers), we ask whether one can find the optimal quantizer(s) from a given set of quantizers. 
To proceed with this framework, we assume that a control system can choose from a given set of quantizers to quantize its measurements and transmit the quantized signal to the controller. 
The set of available quantizers is given a priori along with the cost associated with using each quantizer. 
While the controller aims to minimize the expected quadratic cost, the measurements available to the controller are only the quantized state information. 
It is worth mentioning here that the observation equations are no longer linear due to the quantization process. 
The optimization problem under this framework is a bi-variable decision-making problem where one variable is the control strategy and the other one is the scheduling of the quantizers at each time instance. 

\subsection{Prior and Related Work}

Some of the earlier works on quantization and control can be traced backed to 1970s  \cite{curry1970estimation,schweppe1968recursive, moroney1983issues, delchamps1989extracting}. 
Studies of LQG control under communication constraints with a focus on quantization have been performed in several works such as \cite{borkar1997lqg, tatikonda2000control, tatikonda1998control, tatikonda2004stochastic}, 
 \cite{williamson1989optimal}, \cite{liu1992optimal},\cite{tanaka2016optimal}. 
 For example, \cite{liu1992optimal} considered LQG problems with explicit consideration of the quantization error associated with analog-to-digital implementation. 
 The studies in \cite{williamson1989optimal} and \cite{liu1992optimal}  showed that the optimal controller does \textit{not} exhibit the {separation principle}. 
 It was shown that the optimal controller exhibits a separation principle and that the optimal input-quantizer has to be time-varying with certain specific quantization levels.
  In \cite{tatikonda2004stochastic} the authors provided necessary conditions for the controller to exhibit a separation principle.
   In \cite{borkar1997lqg}, the authors provided a quantization scheme that ensures the existence of a separation principle in the optimal controller. 
   In that work, it was proposed to quantize  a signal which the authors refer to as ``innovations,'' rather than quantizing the state. In this work, we will also adhere to the idea of quantizing the ``innovations" rather than quantizing the state itself.  

Studies on quantization-based control have  also dealt with the stability aspects of the system \cite{delchamps1990stabilizing, wong1997systems,wong1999systems,ishii2003quadratic,li2004robust}. In \cite{wong1997systems} and \cite{wong1999systems}, the authors explicitly considered the issues of quantization, coding and delay. 
The concept of \textit{containability} was used for studying the stability  of linear  quantized systems.
 In \cite{fagnani2002stabilizing}, three quantization schemes (deadbeat, logarithmic, and chaotic) were proposed to ensure \textit{practical stability} of a linear system.
  Optimality of these three quantization schemes was addressed using the notion of {symbolic dynamics}. 
  Symbolic dynamics based analysis was also used in \cite{delchamps1989extracting} for extracting state information from quantized measurements. In \cite{elia2001stabilization} it was shown that the least dense quantizer that quadratically stabilizes a single input linear system is logarithmic. 
  A logarithmic quantizer with finite quantization level can only achieve \textit{practical stabilizability} (a relaxed notion of stabilizability).
   A quantization scheme with time-varying quantization sensitivity was studied in \cite{brockett2000quantized} proving asymptotic stability of the system. 
   In \cite{liberzon2003stabilization} the author derived a relationship between the norm of the transition matrix and the number of values taken by the encoder to ensure global asymptotic stability. 
   The work in \cite{nair2003exponential} addressed the problem of finding the smallest data rate above which exponential stability can be ensured.
    In a more recent work \cite{pearson2017control}, an event-based encoding scheme has been considered.

In the above-mentioned works \cite{borkar1997lqg}--\cite{li2004robust}, the role of quantization has been proven to be  crucial. 
However, for a given control objective, how to select among available quantizers that have an operational 
cost associated with them has not been addressed, and is the subject of this paper.
The problem is similar in spirit to the problem of optimal scheduling of costly sensors for control~\cite{aoki1969optimal,bansal1989simultaneous}, in the sense that measurements are costly and optimal measurements are chosen to maintain an optimal balance between the control performance and observation cost. 
It is however different from these works, in the sense that here we study the effects of quantization in producing the measurements sent to the controller whereas such works do not consider the role of quantization. 

\subsection{Contribution}

The contributions of this work are:

\begin{itemize}

\item We formulate an LQG optimal control problem with a set of \textit{costly quantizers} that quantize the measurements. We seek an optimal controller that minimizes the expected quadratic cost and an optimal selection of the quantizers that determine the quality of the measurements arriving at the controller.

\item We show that quantizing the innovations separates the controller synthesis problem from the \textit{quantizer selection} problem. 
Although the idea of innovation--quantization is presented in \cite{borkar1997lqg},  the initial state of the system was needed to be deterministic in that work. Our framework does not require such an assumption.

\item We study the \textit{optimal controller structure} and show that the controller is of a certainty-equivalence type. The controller gains can be computed offline and the gains do not depend on the parameters of the quantizers.

\item The study of the quantizer-selection reveals that the optimal strategy for the selection of the quantizers can be computed by solving an MDP. Moreover, depending on the information available to the quantizer-selector, the optimal strategy can be computed offline.

\end{itemize}

\subsection{Organization}

 The rest of the paper is organized as follows: in Section \ref{S:prob} we formally define the problem addressed in this paper; Section \ref{S:solution}  provides the structure for the optimal controller and the quantizer selection scheme; numerical examples illustrating the theory are presented in Section \ref{S:simu}.  Finally, we conclude the paper with a summary and some remarks in Section \ref{S:conclusion}.


\section{PROBLEM FORMULATION} \label{S:prob}

Let us consider an LTI discrete-time stochastic system  
\begin{align} \label{E:dyn}
X_{t+1}=AX_t+BU_t+W_t,
\end{align}
where for all $t\in \bbN_0 ~(= \mathbb{N}\cup\{0\})$, $X_t\in \R^n$, $U_t \in \R^m$, $A$ and $ B$ are matrices of compatible dimensions, and $\{W_t\}_{t\in \mathbb{N}_0}$ is an i.i.d noise sequence in $\R^n$ with statistics $W_0 \sim \N(0,\W)$. 
The initial state, $X_0$, is also a Gaussian random variable distributed according to $\N(\mu_0, \Sigma_0)$, and independent of the noise $W_t$ for all $t\in \mathbb{N}_0$. For notational convenience, we will write $X_0=\mu_0+ W_{-1}$ where $W_{-1}\sim \N(0,\Sigma_0)$. Thus, $W_k$ and $W_\ell$ are independent random variable for all $k,\ell = -1,0,1,\ldots$ and $k\ne \ell$.

In this work, we address the quantized feedback LQG (QF-LQG) optimal control problem. 
As shown in Figure \ref{Fig:schematic}, we assume that $M$ quantizers are provided to quantize the state value and transmit the quantized state to the controller. 
The range of the $i$-th quantizer is denoted by  $\Q^i=\{q_1^i,q_2^i,\ldots,q_{\ell_i}^i\}$. 
Associated with the $i$-th quantizer, let $\p^i=\{\p^i_1,\p^i_2,\ldots,\p^i_{\ell_i}\}$ denote a partition in $\R^n$ such that $\p^i_j$ gets mapped to $q^i_j$ for each $j\in \{1,2,\ldots,\ell_i\}$. 
Specifically, one may think of the $i$-th quantizer as a mapping $g_i:\R^n\to \Q^i$ such that $g_i(x)=q^i_j$ if and only if $x\in \p^i_j$. 
Thus, the $i$-th quantizer has $\ell_i$  quantization levels. 
Without loss of generality, we assume that the quantization error covariance decreases from the first quantizer to the $M$-th quantizer, \eg the quantization error covariance is the lowest for the $M$-th quantizer and highest for the first quantizer, and so on.
 Associated with each quantizer, there is an operating cost that must be paid in order to use this quantizer. 
Let $\lambda(\Q^i)=\lambda_i\in \R_+$ denote the cost associated with the $i$-th quantizer\footnote{This framework also extends to the scenario where there is no cost in using a quantizer, i.e., $\lambda_i=0$ for all $i=1,\ldots,M$.}. 
For example, $\lambda_i=\log_2\ell_i$ represents a cost that is proportional to the code-length to encode the output of the quantizer using a simple fixed-length coding scheme.  
In this work, we do not adhere to any specific structure for $\lambda$.
 We assume that the values of $\lambda_i$'s are given to us a priori. 
 If there is a cost for operating the communication channel, that cost can be also incorporated into $\lambda_i$. 
 For example, if every transmission of measurement requires a cost of $\lambda_c$, then $\lambda_i+\lambda_c$ represents the joint quantization and communication cost for using $i$-th quantizer to transmit a measurement.
 Designing such costs in order to regulate the use of the quantizers is an equally interesting problem for the service provider, and that will be addressed elsewhere.
 We will further assume that the communication channel  between each quantizer and the controller always transmits the quantized information without any delay or distortion.
 A discussion on the implications of delay and distortion is provided in Section~\ref{S:Discussion}.

\begin{figure}
\begin{tikzpicture}[scale=0.9, every node/.style={transform shape}]
\draw[fill, color=gray!30] (2.75,-2.75) rectangle +(6,1.7);
\node[anchor=north] at (6.5,-1) {\scriptsize{Set of Quantizers}};
\draw[] (0,0) rectangle +(2,1);
\node at (1,.5) {Controller};
\draw[] (3,0) rectangle +(2,1);
\node at (4,.5) {Plant};
\draw[] (6,0) rectangle +(3,1);
\node at (7.5,.5) {Quantizer Selector};
\draw[] (3.5,-3) --(8,-3);
\draw[fill=white] (3,-2.5) rectangle +(1,1);
\node at (3.5, -2) {$g_1(\cdot)$}; 
\draw[fill=white] (7.5,-2.5) rectangle +(1,1);
\node at (8, -2) {$g_M(\cdot)$}; 
\draw[fill=white] (4.5,-2.5) rectangle +(1,1);
\node at (5, -2) {$g_2(\cdot)$}; 
\draw[thick, dotted] (6,-2) -- (7,-2); 
\draw[->] (6,-3)-- (6, -3.5) -- (1,-3.5)-- (1,0);
\draw[->] (2,.5) --(3,.5);
\node[anchor=south] at (2.5,.5) {$U_t$};
\draw[->] (5,.5) --(6,.5);
\node[anchor=south] at (5.5,.5) {$X_t$};
\draw (3.5,-2.5)-- (3.5,-3);
\draw (5,-2.5)-- (5,-3);
\draw (8,-2.5)-- (8,-3);
\draw (3.5,-1.5)-- (3.5,-1);
\draw (5,-1.5)-- (5,-1);
\draw (8,-1.5)-- (8,-1);
\draw[->] (5.5, .5)--(5.5, -.5) -- (5.1,-.9);
\draw[->] (7.5, 0)--(7.5, -.75) -- (5.5,-.75);
\node[anchor =west] at (7.5, -0.5) {$\theta_t$};
\node[anchor=north] at (5,-3.5) {Quantized Signal $Y_t$};
\draw[->] (6, -3.5) -- (9.5,-3.5)-- (9.5,0.5)--(9,0.5);
\end{tikzpicture}
\caption{Schematic diagram of the system where the gray block contains the set of quantizers and the desirable quantizer ($g_i$) is selected by the quantizer selector variable $\theta$. 
} \label{Fig:schematic}
\end{figure}
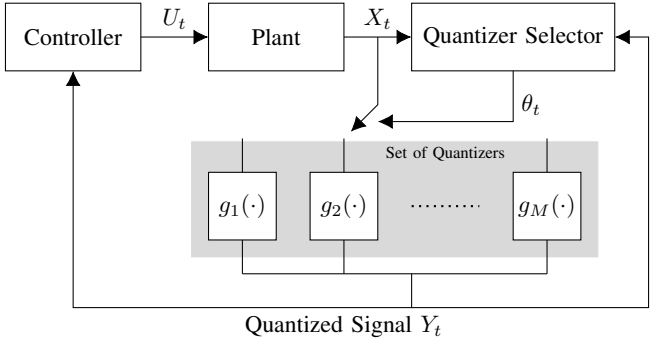


The objective is to minimize a finite-horizon performance index that takes into account the quantization cost. 
Contrary to the existing literature  on quantization-based LQG \cite{borkar1997lqg}-\cite{fischer1982optimal}, in our case there are two decision-makers instead of one: one (the controller) decides the input ($\{U_t\}_{t\in \bbN_0}$) to apply to the system, and the other (the quantizer-selector) decides the quality of the measurements (quantized state values) which are transmitted to the controller.
That is, the proposed work does not fit with the previous quantization literature that focuses primarily on stability issues 
or the \textit{design} of the optimal quantization scheme to maintain performance. 

We introduce a new decision variable $\theta^i_t$ for the quantizer-selector in the following way:
\begin{align*}
\theta_t^i =\begin{cases} &1, ~~~i\text{-th quantizer is used at time } t,\\
 &0,~~~\text{otherwise.}
\end{cases}
\end{align*}
Let us denote the vector $\theta_t\triangleq[\theta_t^1,\theta_t^2,\ldots,\theta_t^M]\T \in \{0,1\}^M$, that characterizes the decision of the quantizer-selector at time $t$. 
We enforce the quantizer-selector to select only one quantizer at any time instance, and hence we have $\sum_{i=1}^M\theta_t^i=1$  for all $t\in \bbN_0$.

The measurement available to the controller at time $t$ is represented as $Y_t=\sum_{i=1}^Mg_i(X_t)\theta_t^i=g_j(X_t)\theta^j_t$, if the $j$-th quantizer is selected at time $t$.\footnote{In the following analysis we shall quantize $W_t$  instead of $X_t$. This will help us preserve the separation-principle for the optimal controller as shown in \cite{borkar1997lqg}.} 
Let us also introduce the sets $\X_t\triangleq\{X_0,X_1,\ldots,X_t\}$, $\Y_t\triangleq \{Y_0,Y_1,\ldots,Y_t\}$, $\U_t\triangleq\{U_0,U_1,\ldots,U_t\}$ and $\Theta_t\triangleq\{\theta_0,\theta_1,\ldots,\theta_t\}$ to be the state history, measurement history, control history and quantization-selection history respectively. 
For convenience, we will use the notation $\U$ for $\U_{T-1}$ and, likewise, $\Theta$ for $\Theta_{T-1}$ to denote the history of the entire horizon $[0,T]$.

The information available to the controller at time $t$ is $\I_t^c=\{\Y_t,\U_{t-1},\Theta_t\}=\I_{t-1}^c\cup\{Y_t,U_{t-1},\theta_t\}$ with $\I_0^c=\{Y_0, \theta_0\}$. 
It should be noted that $\I^c_t$ depends on $\Theta_t$ through $\Y_t$. 
In classical optimal LQG control, the information available to the controller is not decided by any active decision maker, unlike the situation here.
An admissible control strategy at time $t$ is a measurable function from the Borel $\sigma$-field generated by $\I_t^c$ to $\R^m$. Let us denote such strategies by $\gamma^u_t(\cdot)$ and the space they belong to by $\Gamma^u_t$.
On the other hand, the information available to the quantizer-selector at time $t$ is $\I_t^q=\{\X_t,\Y_{t-1},\U_{t-1},\Theta_{t-1}\}=\I_{t-1}^q\cup\{X_t,Y_{t-1},U_{t-1},\theta_{t-1}\}$ with $\I_0^q=\{X_0\}$. 
The admissible strategies for the selection of the quantizers are measurable functions from the Borel $\sigma$-field generated by $\I_t^q$ to $\{0,1\}^M$. 
Let us denote such strategies by $\gamma^\theta_t(\cdot)$, and the space they belong to by $\Gamma^\theta_t$.
For brevity, often we will use $\gamma^u_t$ instead of $\gamma^u_t(\cdot)$ or $\gamma^u_t(\I_t^c)$, and $\gamma^\theta_t$ instead of $\gamma^\theta_t(\cdot)$ or $\gamma^\theta_t(\I^q_t)$.
 Let $\gamma^\Theta$ denote the entire sequence $\{\gamma^\theta_0,\gamma^\theta_1,\ldots,\gamma^\theta_{T-1}\}$ and let $\Gamma^\Theta$ denote the space where $\gamma^\Theta$ belongs to. 
 Likewise, $\gamma^\U$ and $\Gamma^\U$ are defined similarly.
The sequence of decision-making within one time instance is then as follows: \\
$\cdots\to  \I_t^q\overset{\gamma^\theta_t}{\rightarrow}  \theta_t\to Y_t \to \I_t^c\overset{\gamma^u_t}{\to} U_t\to X_{t+1}\to \I_{t+1}^q\to \cdots$.


The cost function to be minimized cooperatively by the controller and the quantizer-selector is a finite horizon quadratic criterion,  given as
\begin{align} \label{E:cost}
J(\U,\Theta)=\E\left[\sum_{t=0}^{T-1}(X_t\T Q_1X_t+U_t\T RU_t+\theta_t\T \Lambda) +X_T\T Q_2X_T \right],
\end{align}
where $\Lambda=[\lambda_1,\lambda_2,\ldots,\lambda_M]\T $, $Q_1,Q_2 \succeq 0, R\succ 0$, $\U=\gamma^\U(\I^c)=\{\gamma^u_0(\I_0^c), \gamma^u_1(\I_1^c),\ldots,\gamma_{T-1}^u(\I_{T-1}^c)\}$ and  $\Theta=\gamma^\Theta(\I^q)=\{\gamma^\theta_0(\I_0^q), \gamma^\theta_1(\I_1^q),\ldots,\gamma_{T-1}^\theta(\I_{T-1}^q)\}$.
We seek to find the optimal strategies $\gamma^{\U*}=\{\gamma_0^{u*},\gamma_1^{u*},\ldots,\gamma_{T-1}^{u*}\}$ and $\gamma^{\Theta*}=\{\gamma^{\theta*}_0,\gamma^{\theta *}_1,\ldots,\gamma^{\theta *}_{T-1}\}$ that minimize \eqref{E:cost}. 
To this end, we will also rewrite \eqref{E:cost} in terms of $\gamma^\U$ and $\gamma^\Theta$ as
\begin{align} \label{E:cost2}
J(\gamma^\U,\gamma^\Theta)=\E\Big[&\sum_{t=0}^{T-1}(X_t\T Q_1X_t+U_t\T RU_t+\theta_t\T \Lambda) +X_T\T Q_2X_T\nonumber\\&~~~|~U_t=\gamma^u_t(\I^c_t), \theta_t=\gamma^\theta_t(\I^q_t) \Big].
\end{align}
We emphasize that $\I^q_t$ contains the state value $X_t$ for all $t$. 
Such an information structure will be refereed to as \textit{perfect
measurement quantizer selection}. 
In Section \ref{S:openloop}, we will consider a different information structure for $\I^q_t$ where $X_t$ is not present, and  $\I^q_t$ contains the same quantized signals as the controller information. 
Such information pattern will be referred to as \textit{quantized measurement quantizer selection}.
The perfect measurement  quantizer selection leads to a MDP formulation, which is more computationally expensive to solve compared to the quantized measurement quantizer selection scenario which can be solved through linear programming.
As will be discussed in detail later, the available information for selecting the quantizers in the quantized measurement case is a subset of the information available for the perfect measurement case.
Thus, the perfect measurement scenario results in better performance (i.e., lower cost) than the quantized measurement case, albeit at an expense of higher computation complexity.

\section{Optimal Control and Quantization Selection} \label{S:solution}

\subsection{Perfect Measurement Quantizer Selection} \label{S:closeloop}

In this section we find the optimal $\gamma^{\U*}$ and $\gamma^{\Theta *}$ that minimize the cost function \eqref{E:cost2} amongst all admissible strategies, that is,
\begin{align}\label{E:argmin}
(\gamma^{\U*},\gamma^{\Theta*})=\underset{\gamma^\U\in \Gamma^\U,\gamma^\Theta\in \Gamma^\Theta}{\arg\min}J(\gamma^\U,\gamma^\Theta).
\end{align}
Before proceeding further to solve \eqref{E:argmin}, let us specify the input for the quantization process, since it will play a crucial role in the following analysis. Unlike other quantized feedback based control approaches \cite{williamson1989optimal}, \cite{liu1992optimal}, we will quantize $W_{t-1}$ instead of $X_t$ at time $t$. 
 Note that, $W_{t-1}$ can be readily computed from the values of $X_t, X_{t-1}, U_{t-1}$ that are included in $\I^q_t$. 
In the existing literature \cite{kostina2019rate,nair2004stabilizability} it has been shown that quantizing the state leads to a problem that is intractable, whereas in \cite{borkar1997lqg} the utility of noise quantization has been proposed. 
 The schematic for the noise quantization based framework is presented in Figure \ref{Fig:update_schematic}.
 
\begin{figure}
\begin{tikzpicture}[scale=0.9, every node/.style={transform shape}]
\centering
\draw[fill, color=gray!30] (2.75,-2.75) rectangle +(6,1.7);
\node[anchor=north] at (6.5,-1) {\scriptsize{Set of Quantizers}};
\draw[] (0,0) rectangle +(2,1);
\node at (1,.5) {Controller};
\draw[] (3,0) rectangle +(2,1);
\node at (4,.5) {Plant};
\draw[] (6,1.5) rectangle +(3,1);
\node at (7.5,2) {Quantizer Selector};
\draw[] (6,0) rectangle +(3.2,1);
\node at (7.6,.5) {Innovation extraction};
\draw[] (3.5,-3) --(8,-3);
\draw[fill=white] (3,-2.5) rectangle +(1,1);
\node at (3.5, -2) {$g_1(\cdot)$}; 
\draw[fill=white] (7.5,-2.5) rectangle +(1,1);
\node at (8, -2) {$g_M(\cdot)$}; 
\draw[fill=white] (4.5,-2.5) rectangle +(1,1);
\node at (5, -2) {$g_2(\cdot)$}; 
\draw[thick, dotted] (6,-2) -- (7,-2); 
\draw[->] (6,-3)-- (6, -3.5) -- (1,-3.5)-- (1,0);
\draw[->] (2,.5) --(3,.5);
\node[anchor=south] at (2.5,.5) {$U_t$};
\draw[->] (5,.5) --(6,.5);
\node[anchor=north] at (5.5,.5) {$X_t$};
\draw (3.5,-2.5)-- (3.5,-3);
\draw (5,-2.5)-- (5,-3);
\draw (8,-2.5)-- (8,-3);
\draw (3.5,-1.5)-- (3.5,-1);
\draw (5,-1.5)-- (5,-1);
\draw (8,-1.5)-- (8,-1);
\draw[->] (5.5, .5)--(5.5, 2) -- (6,2);
\draw[->] (7.5, 0)--(7.5, -.6) -- (8,-1);
\node[anchor =west] at (7.5, -0.3) {$W_{t-1}$};
\node[anchor=west] at (3,-3.75) {Quantized Signal $\hat{w}_{t-1}(\theta_t)$};
\draw[->] (9,2)--(9.5,2)--(9.5,-.7)--(7.75,-.7);
\node[anchor=south] at (9,-.7) {$\theta_t$};
\draw[->] (6, -3.5) -- (9.75,-3.5)-- (9.75,2.25)--(9,2.25);
\end{tikzpicture}
\caption{Schematic diagram of the (perfect measurement quantizer selection) system, 
where the gray block contains the set of quantizers and the desirable quantizer ($g_i$) is selected by the quantizer selector variable $\theta$.} \label{Fig:update_schematic}
\end{figure}
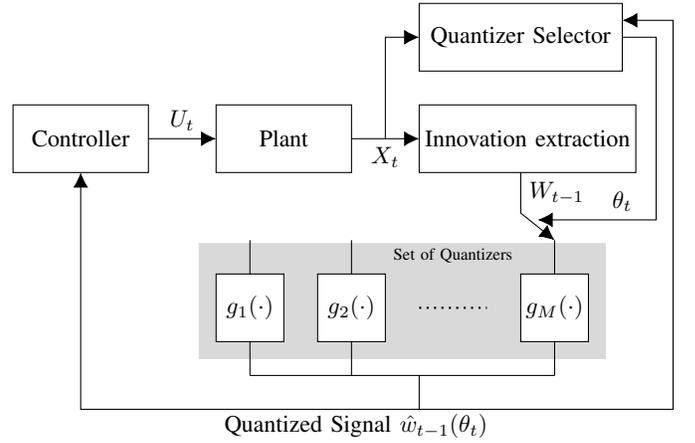

 Let $g_i(W_{t-1}) \in \Q^i$ denote the quantized version of  $W_{t-1}$ if the $i$-th quantizer is selected. 
 Therefore, the  quantized information sent to the controller is
 \begin{align} \label{E:measure}
Y_t=\sum_{i=1}^Mg_i(W_{t-1})\theta^i_t
 \end{align}
 
 Note that $Y_t$ is a random variable taking values in the discrete set $\cup_{i=1}^M\Q^i$ with $\mathsf{P}(Y_t=q^i_j)=\mathsf{P}(W_{t-1}\in \p^i_j)$. 
  Let us now define $\hat w_t^i\triangleq \E[W_t~|~g_i(W_t)]$. 
It should be noted that $\hat w^i_t$ is a function of $W_t$ although it is not explicitly expressed as such.
It is the estimate of the noise $W_{t}$ given the output from the $i$-th quantizer.
If the quantizers are optimal \cite{gersho2012vector}, the centroid condition\footnote{If $X$ is a random variable quantized by $\Q^i$, then $\Q^i$ is optimal if $q^i_j=\E[X|X\in \p^i_j]$ for all $ j=1,2,\ldots,\ell_i$.} of optimality implies: 
\begin{align} \label{E:indicatorw}
\hat{w}^i_t=\sum_{i=1}^{\ell_i}q^i_j1_{\p^i_j}(W_t),
\end{align}
where $1_S(\cdot)$ is the indicator function of the set $S$. 
Equation \eqref{E:indicatorw} reflects the fact that the optimal estimate\footnote{Unbiased estimation with minimum covariance.} of the noise would be the quantized value sent by the quantizer when the quantizer satisfies a certain optimality condition \cite{linde1980algorithm}. 
However, in our study, the quantizers need not be optimal. 
If the $i$-th quantizer has been used to quantize $W_t$ and the quantized value is $q^i_j=g_i(W_t)$, then we know that $W_t\in \p^i_j$.

  One can verify that $\E[\hat w^i_t]=\E\left[\E[W_t~|~g_i(W_t)]\right]=\E[W_t]=0$ for all $i=1,2,\ldots, M$ and $t =-1,0,\ldots,T-1$. 
It follows immediately from the definitions of $\I^c_t$ and $\I^q_t$ that
\begin{align} \label{E:w}
\E[W_t|\I^c_t]=\E[W_t|\I^q_t]=0,
\end{align} 
and
\begin{align}
 \E[W_t|\I^c_{t+1}]&=\E[W_t|\I^c_t,Y_{t+1},U_t,\theta_{t+1}]\nonumber\\&\stackrel{(a)}{=}\E[W_t|Y_{t+1},\theta_{t+1}]\stackrel{(b)}{=}\sum_{i=1}^M\theta^i_{t+1}\E[W_t|g_i(W_{t})]\nonumber \\&=\sum_{i=1}^M\theta^i_{t+1}\hat w^i_t\triangleq\hat w_t(\theta_{t+1}),
\end{align}
where (a) follows from the fact that $W_t$ is independent of $U_t$ and $\I^c_t$, and (b) follows from the facts that $\theta_{t+1}\in \{0,1\}^M$, $\sum_{i=1}^M\theta^i_{t+1}=1$ and  $Y_{t+1}=\sum_{i=1}^M\theta^i_{t+1}g_i(W_t)$.
To be more precise, we should have written $\hat{w}_t(\theta_{t+1})$ as $\hat{w}_t(\theta_{t+1},W_t)$ since $\hat{w}_t$ depends on the noise realization. 
Often times, for notational brevity, we will simply use $\hat{w}_t$ and suppress the $\theta_{t+1}$ argument.
Although we shall use $\hat{w}_t$ or $\hat{w}_t(\theta_{t+1})$ instead of $\hat{w}_t(\theta_{t+1},W_t)$, we must keep in mind that $\hat{w}_t$ is a random variable that depends on $W_t$.

Having a quantizer with just one quantization level is equivalent to operating in open-loop since no information about the input signal is retained in the quantized signal. 
Therefore, in our study, one can include a hypothetical quantizer $\Q^0$ with a single quantization level and cost $\lambda_0=0$ to account for the possibility to remain open-loop at any time whenever this quantizer is selected.

Let us also define $\hat X_t=\E[X_t|\I^c_{t-1}]$, which will be referred to as the \textit{prediction} of $X_t$, and $\tilde{X}_t=\E[X_t|\I^c_t]$ which  will be referred to as the \textit{filtered} version of $X_t$. 

Using \eqref{E:w} and the fact that $U_t$ is $\I^c_t$-measurable, one can write
\begin{align}
\hat X_{t+1}=A\tilde{X}_t+BU_t,
\end{align} 
where 
\begin{align} \label{E:filter}
\tilde{X}_t=&\E[X_t|\I^c_t]=\E[AX_{t-1}+BU_{t-1}+W_{t-1}|\I^c_t]\nonumber\\
=&A\tilde X_{t-1}+BU_{t-1}+\hat{w}_{t-1}(\theta_t)\nonumber\\
=&\hat X_t +\hat{w}_{t-1}(\theta_t).
\end{align}

Note that $\tilde{X}_t$ depends on $\Theta_t$ and $\hat X_t$ depends on $\Theta_{t-1}$. In \eqref{E:filter}, $\hat w_{t-1}(\theta_t)$ is the only term that depends on $\theta_t$.

Let us define the error  $\Delta_t=X_t-\tilde{X}_t$. It follows that,
\begin{align}\label{E:del}
\Delta_{t+1}=&A\Delta_t+W_t-\hat{w}_t \nonumber\\
=&A^{t+1}\Delta_0+\sum_{k=0}^tA^{t-k}(W_k-\hat{w}_k),
\end{align}
where $\Delta_0=W_{-1}-\hat{w}_{-1}(\theta_0)$. 
Therefore, the state estimation error, $\Delta_t$,  depends on $\{\theta_0, \ldots, \theta_t\}$ through the variables $\{\hat{w}_{-1},\ldots,\hat w_{t-1}\}$. It does not depend on the control strategy $\gamma^\U$. 
This implies a separation structure between the controller and the quantizer-selection. 
In the following, we will formally show the emergence of a separation-principle for this problem.

Associated with the cost function \eqref{E:cost2}, let us define the value function as follows:
\begin{align}\label{E:preV}
V_k(x)=&\min_{\{\gamma^u_t\}_{t=k}^{T-1},\{\gamma^\theta_t\}_{t=k}^{T-1}}\E\Big[\sum_{t=k}^{T-1}(X_t\T Q_1X_t+U_t\T RU_t+\theta_t\T \Lambda) \nonumber\\&~~~+X_T\T Q_2X_T~\big|~U_t=\gamma^u_t(\I^c_t), \theta_t=\gamma^\theta_t(\I^q_t),\nonumber\\&~~~~~~~~~~~~~~~~~~ X_k=x,t=k,\ldots,T-1 \Big].
\end{align}

By the optimality principle,
\begin{align} \label{E:V}
V_k(x)=&\min_{\gamma^u_k\in \Gamma^u_k,\gamma^\theta_k\in \Gamma^\theta_k}\E\Big[(X_k\T Q_1X_k+U_k\T RU_k+\theta_k\T \Lambda) \nonumber\\&+V_{k+1}(X_{k+1})~\big|~U_k=\gamma^u_k(\I^c_k), \theta_k=\gamma^\theta_k(\I^q_k),\nonumber\\&~~~~~~~~~~~~~~~~~~~~~X_k=x \Big].
\end{align}

If $\gamma^{u*}_k$ and $\gamma^{\theta*}_k$ minimize the right-hand-side of \eqref{E:V}, then $U_k^*=\gamma^{u*}_k(\I^c_k)$ and $\theta_k^*=\gamma^{\theta *}_k(\I^q_k)$.
From \eqref{E:preV}, we also have that
\begin{align} \label{E:exVal}
\min_{\gamma^\U\in \Gamma^\U,\gamma^\Theta\in \Gamma^\Theta}J(\gamma^\U,\gamma^\Theta)=\E[{V_0(X_0)}],
\end{align}
where the expectation in \eqref{E:exVal} is taken over the random variable $X_0$.
In order to maintain notational brevity in the subsequent analysis, we will write $V_k(x)$ as follows:
\begin{align*}
V_k(x)=\min_{\gamma^u_k,\gamma^\theta_k}\E_x\Big[&(X_k\T Q_1X_k+U_k\T RU_k+\theta_k\T \Lambda) \nonumber\\&+V_{k+1}(X_{k+1})~\big],
\end{align*}
where  $\E_x[\cdot]$ in this context will denote the conditional expectation given the event $X_k=x$, and $U_k$ and $\theta_k$ are implicitly assumed to be of the form $U_k=\gamma^u_k(\I^c_k), \theta_k=\gamma^\theta_k(\I^q_k)$ for some $\gamma^u_k\in \Gamma^u_k$ and $\gamma^\theta_k\in \Gamma^\theta_k$. 
The following  Theorem characterizes the optimal policy $\gamma^{u*}_k(\cdot)$ for all $k=0,1,\ldots,T-1$.\\

\begin{thm} \label{T:optcont}
Given the information $\I^c_k$ to the controller at time $k$, the optimal control policy $\gamma^{u*}_k:\I^c_k\to\R^m$ that minimizes the right-hand-side of \eqref{E:V} has the following structure
\begin{align}
U^*_k=\gamma^{u*}_k(\I^c_k)=-L_k\E[{X}_k|\I^c_k],
\end{align}
where for all $k =0,1,\ldots,T-1$, $L_k$ and $P_k$ are obtained by
\begin{subequations}
\begin{align}
\begin{split}
L_k=(R+B\T P_{k+1}B)^{-1}B\T P_{k+1}A, \label{E:eqLk}
\end{split}\\
\begin{split}
P_k= Q_1+A\T P_{k+1}A-L_k\T (R+B\T P_{k+1}B)L_k, \label{E:eqPk}
\end{split}\\
\begin{split}
P_T=Q_2.
\end{split}
\end{align}
\end{subequations}
\end{thm}
\vspace{10 pt}

\begin{proof}
The proof of this theorem is based on the dynamic programming principle. 
If there exist value functions $V_k(\cdot)$ for all $k=0,1,\ldots,T$ that satisfy \eqref{E:V}, then the optimal control $U_k^*$ and the optimal quantizer selection $\theta_k^*$ are obtained by the policies $\gamma^{u*}_k$ and $\gamma^{\theta*}_k$ that minimize \eqref{E:V}. 

Let us assume that the value function at time $k =0,1,\ldots,T-1$ is of the form:
\begin{align} \label{E:hypoV}
V_k(x)=x\T P_kx+C_k+r_k,
\end{align}
where $P_k$ is as  in \eqref{E:eqPk}, and for all $k=0,1,\ldots,T-1$,
\begin{align} \label{E:ck_old}
C_k=\min_{\{\gamma^\theta_t\}_{t=k}^{T-1}}\E\left[\sum_{t=k}^{T-1}\Delta_t\T N_t\Delta_t+\theta_t\T \Lambda\right],
\end{align}
 where $N_k \in \R^{n\times n}$ and $r_k\in \R$ are given by
 \begin{subequations}
\begin{align}
N_k=&L_k\T (R+B\T P_{k+1}B)L_k, \label{E:eqNk} \\ 
r_k=&r_{k+1}+\tr(P_{k+1}\W), \label{E:eqrk} \\
r_T=&0.
\end{align}
\end{subequations}

{Therefore, neither $N_k$, nor $r_k$ depends on the past (or future) decisions on the control or quantizer-selection. Therefore, these quantities can be computed offline.}
Moreover, one  notices from \eqref{E:del} that $\Delta_k$ does not depend on the past control history $\U_k$ and it is solely characterized by $\Theta_k$. Thus, $C_k$ does not depend on the control action $U_k$.
Equation \eqref{E:ck_old} can thus be re-written as
\begin{align*}
C_{k}=\min_{\gamma^\theta_k}\E\left[\Delta_k\T N_k\Delta_k+\theta_k\T \Lambda+C_{k+1} \right].
\end{align*}

From the definition of $V_k(\cdot)$ in \eqref{E:preV}, we can write
$V_T(x)=x\T Q_2x=x\T P_Tx$ for all $x\in \R^n$.
 Next, we verify that $V_{T-1}(x)$ is of the form \eqref{E:hypoV}. Note that
\begin{align} \label{E:eqVT-1}
V_{T-1}(x)=\min_{\gamma^u_{T-1},\gamma^\theta_{T-1}}\E_x\Big[&X_{T-1}\T Q_1X_{T-1} +U_{T-1}\T RU_{T-1}\nonumber\\&+\theta_{T-1}\T \Lambda+X_T\T P_TX_T \Big].
\end{align}
Substituting the equation $X_T=AX_{T-1}+BU_{T-1}+W_{T-1}$, and after some simplifications, yields
\begin{align*}
V_{T-1}(x)=&\min_{\gamma^u_{T-1},\gamma^\theta_{T-1}}\E_x\Big[\|U_{T-1}+L_{T-1}X_{T-1}\|^2_{(R+B\T P_TB)}\\&+X_{T-1}\T P_{T-1}X_{T-1}+\theta_{T-1}\T \Lambda+\tr(P_T\W) \Big],
\end{align*}
where $\|\cdot\|^2_K$ denotes a weighted norm with $K$ being the weight matrix.
In the previous expression, $\|U_{T-1}+L_{T-1}X_{T-1}\|^2_{(R+B\T P_TB)}$ is the only term that depends on $U_{T-1}$. 
Therefore, we seek $\gamma^u_{T-1}:\I^c_{T-1}\to \R^m$  that minimizes the mean-square error $\E_x\left[\|U_{T-1}+L_{T-1}X_{T-1}\|^2_{(R+B\T P_TB)}\right]$. 
Thus, $U_{T-1}$ is a minimum mean squared estimate of $L_{T-1}X_{T-1}$. Hence,
\begin{align}
U_{T-1}^*=\gamma^{u*}_{T-1}(\I^c_{T-1})=-L_{T-1}\E[X_{T-1}|\I^c_{T-1}].
\end{align} 
After substituting the optimal $U^*_{T-1}$ in \eqref{E:eqVT-1}, we obtain
\begin{align*}
V_{T-1}(x)=\min_{\gamma^\theta_{T-1}}\E_x\Big[&\|X_{T-1}-\tilde X_{T-1}\|^2_{N_{T-1}}+\theta_{T-1}\T \Lambda\\& +\tr(P_T\W)+X_{T-1}\T P_{T-1}X_{T-1}\Big].
\end{align*}
Given the event $X_{T-1}=x$,  the above expression of $V_{T-1}(x)$ can be simplified as follows
\begin{align*}
V_{T-1}(x)=&\min_{\gamma^\theta_{T-1}}\E\left[\Delta_{T-1}\T N_{T-1}\Delta_{T-1}+\theta_{T-1}\T \Lambda\right]\\&+x\T P_{T-1}x+\tr(P_T\W).
\end{align*}
Therefore, using the definitions of $C_{T-1}$ and $r_{T-1}$, we obtain
\begin{align*}
V_{T-1}(x)=C_{T-1}+x\T P_{T-1}x+r_{T-1}.
\end{align*}
Thus, $V_{T-1}$ is of the form \eqref{E:hypoV}. 
Let us now assume that \eqref{E:hypoV} is true for some $k+1$. Then
\begin{align*}
V_k(x)=\min_{\gamma^u_k,\gamma^\theta_k}\E_x\Big[&(X_k\T Q_1X_k+U_k\T RU_k+\theta_k\T \Lambda) \nonumber\\&+V_{k+1}(X_{k+1})~\big]\\
=\min_{\gamma^u_k,\gamma^\theta_k}\E_x\Big[&(X_k\T Q_1X_k+U_k\T RU_k+\theta_k\T \Lambda) \nonumber\\+&X_{k+1}\T P_{k+1}X_{k+1}+r_{k+1}+C_{k+1}(\Delta_{k+1})~\big].
\end{align*} 

Using \eqref{E:dyn}, and after some simplifications, it follows that
\begin{align} \label{E:eqVk}
V_k(x)=&\min_{\gamma^u_{k},\gamma^\theta_{k}}\E_x\Big[\|U_{k}+L_{k}X_{k}\|^2_{(R+B\T P_{k+1}B)}+X_{k}\T P_{k}X_{k}\nonumber \\&+\theta_{k}\T \Lambda+\tr(P_{k+1}\W)+r_{k+1}+C_{k+1} \Big].
\end{align}
By the principle of minimum-mean-square estimate, the optimal $\I_k^c$-measurable control $U_k^*$ that minimizes  $\E\left[ \|U_{k}+L_{k}X_{k}\|^2_{(R+B\T P_{k+1}B)}\right]$ is given by
\begin{align} \label{E:u*}
U_k^*=\gamma^{u*}_k(\I_k^c)=-L_k\E\left[X_k|\I^c_k\right].
\end{align} 
After substituting the optimal control in \eqref{E:eqVk}, yields
\begin{align*}
V_k(x)=&x\T P_kx+\min_{\gamma^\theta_k}\E_x\Big[\Delta_k\T (L_k\T (R+B\T P_{k+1}B)L_k)\Delta_k\\&+\theta_k\T \Lambda+C_{k+1}\Big]+\tr(P_{k+1}\W)+r_{k+1}\\
=&x\T P_kx+\min_{\gamma^\theta_k}\E\Big[\Delta_k\T N_k\Delta_k+\theta_k\T \Lambda+C_{k+1}\Big]+r_k\\
=&x\T P_kx+C_k+r_k.
\end{align*}
Thus, the value function is indeed of the form \eqref{E:hypoV}, and the optimal control at time $k=0,1,\ldots,T-1$ is given in \eqref{E:u*}.
\end{proof}

From Theorem \ref{T:optcont}, the optimal control is linear in  $\tilde{X}_k$. 
The optimal gain is $-L_k$, which can be computed offline without knowledge of $\gamma^{\Theta*}$. 
The effect of $\gamma^{\Theta*}$ on $\gamma^{\U*}$ is through the term $\tilde{X}_k$, which can be computed online using \eqref{E:filter}.
From \eqref{E:hypoV}, we have
\begin{align*}
V_0(X_0)=X_0\T P_0X_0+C_0+r_0.
\end{align*}
Thus,
\begin{align*}
\min_{\gamma^\U\in \Gamma^\U,\gamma^\Theta\in \Gamma^\Theta}J(\gamma^\U,\gamma^\Theta)=&\E[{V_0(X_0)}]\\=&\tr(P_0\Sigma_0)+r_0+\E[C_0],
\end{align*}
where 
\begin{align} \label{E:ck}
C_0=&\min_{\{\gamma^\theta_t\}_{t=0}^{T-1}}\E\left[\sum_{t=0}^{T-1}\Delta_t\T N_t\Delta_t+\theta_t\T \Lambda\right]\nonumber \\
=&\min_{\{\gamma^\theta_t\}_{t=0}^{T-1}}\E\left[\sum_{t=0}^{T-1}c_t(\Delta_t,\theta_t)\right],
\end{align}
where $c_t(\Delta_t,\theta_t)=\Delta_t\T N_t\Delta_t+\theta_t\T \Lambda$.

Next, we study the optimal quantizer-selection policy $\gamma^{\Theta*}.$
The key insight here is that equation \eqref{E:ck} is reminiscent of a Markov-decision-process (MDP). 
However, it is not a standard MDP. 
Unlike a standard MDP, $\Delta_{t}$ can be deterministically characterized by the information $\I^q_t$ and $\theta_{t}$.
This is due to fact that, before selecting $\theta_t$, the disturbance $W_{t-1}$ is known. 
In the following, we reduce \eqref{E:ck} to a standard MDP problem in some new state-space $\s$ and action-space $\A$.

Let us consider the state at time $t$ to be $S_t=[\Delta_{t-1}\T ,W_{t-1}\T ]\T \in \R^{2n}$. 
The action space is the set of all canonical basis vectors of $\R^M$, \eg $\A=\{b_1,b_2,\ldots,b_M\}$ where $b_i\in \R^M$ is the $i$-th basis vector.
From \eqref{E:del} it follows that the dynamics of $S_t$ is given by
\begin{align}\label{E:s}
S_{t+1}=&\begin{bmatrix}
A\Delta_{t-1}+W_{t-1}-\hat w_{t-1}(\theta_t)\nonumber\\
W_{t}
\end{bmatrix}\\&=HS_t-\begin{bmatrix}
\sum_{i=1}^M\theta^i_t\hat w_{t-1}^i\\ 0
\end{bmatrix} +\begin{bmatrix}
0\\W_{t}
\end{bmatrix}\nonumber\\
&=HS_t-\begin{bmatrix}
G_{t}(S_{t})\theta_t\\0
\end{bmatrix}+\begin{bmatrix}
0\\W_t
\end{bmatrix}\nonumber\\
&\triangleq f(t,S_t,\theta_t,W_t),
\end{align}
where $H=\begin{bmatrix}
A~~&{I}\\{0}~~&{0}
\end{bmatrix}$, $G_{t}(S_t)=[\hat w_{t-1}^1, \hat w_{t-1}^2,\ldots,\hat w_{t-1}^M]$. $\hat w_{t-1}^i=\E[W_{t-1}~|~g_i(W_{t-1})]$ is a function of the state $S_t$.
The initial condition is
\begin{align*}
 S_0=\begin{bmatrix}
0\\ W_{-1}
\end{bmatrix}\sim \N\left(0,\begin{bmatrix}
{0} &0 \\0 &\Sigma_0
\end{bmatrix}\right),
\end{align*}
and the observation equation  is
\begin{align*}
Z_t=W_t=[{0}~~ {I}]S_t.
\end{align*} 
Therefore, $S_t$ can be constructed from the histories $\X_{t-1}$, $\U_{t-1}$ and $\Theta_{t-1}$, all of which are available in $\I^q_t$.
Also, note that $\Delta_{t}=[{I}~~ {0}]S_{t+1}$ and therefore, $c_t(\Delta_t,\theta_t)$ can be re-written as
\begin{align*}
c_t(\Delta_t,\theta_t)=\tilde{c}_t(S_{t+1},\theta_t)=S_{t+1}\T \tilde{N}_tS_{t+1}+\theta_t\T \Lambda,
\end{align*}
where $\tilde{N}_t=\begin{bmatrix}
N_t & {0}\\ {0} & {0}
\end{bmatrix} $.
Thus, \eqref{E:ck} can be cast as an MDP problem over a continuous state-space $\s=\R^{2n}$ and finite action space $\A$.
Specifically,
we have the following theorem.

 
\begin{thm} \label{T:optquant}
The optimal quantizer selection can be found by solving the  MDP $(\s,\A,{\sf{P}},\tilde{c})$ in \eqref{E:mdp} with state-space $\s= \R^{2n}$ and action-space $ \A=\{b_1,b_2,\ldots,b_M\}$,
\begin{align}\label{E:mdp}
&\min_{\gamma^\Theta\in \Gamma^\Theta}\E\left[\sum_{t=0}^{T-1}\tilde c({S_{t+1}},\theta_t)\right],\\ \nonumber
\text{s.t.} ~~~~& S_{t+1}=f(t,S_t,\theta_t,W_t),
\end{align}
with corresponding transition probabilities  ${\sf P}(S_{t+1}|S_t,\theta_t)\sim \N(\mu(S_t,\theta_t),\Sigma_t)$ where $\mu(S_t,\theta_t)=HS_t-\begin{bmatrix}
G_{t}(S_{t})\theta_t\\0
\end{bmatrix}$ and $\Sigma_t=\begin{bmatrix}
{0} & {0}\\ {0} & \W
\end{bmatrix}$.\\
\end{thm}

\begin{proof}
The proof of this Theorem follows from the fact that \eqref{E:ck} is equivalent to \eqref{E:mdp} under the dynamics of $S_t$ given in \eqref{E:s}.
\end{proof}
\vspace*{5 pt}

Despite  the quadratic nature of the cost and Gaussian distribution of the noise, a closed-form solution to the above MDP is not possible due to the non-linear dynamics  in \eqref{E:mdp}.

Due to the Markovian structure of the problem, we can restrict ourselves to the space of Markovian policies, \eg $\gamma^\theta_k:S_k\to \A$ instead of $\gamma^\theta_k:\I^q_k\to \A$. 
The space of all Markovian strategies at time $k$ (entire horizon) are denoted by $\Gamma^{\theta,\mathcal M}_k \subseteq \Gamma^\theta_k$ ($\Gamma^{\Theta, \mathcal M}\subseteq \Gamma^{\Theta}$).
%
With a slight abuse of notation, let us use $C_k(S_k)$ to denote the optimal cost-to-go from time $k$ for the MDP in \eqref{E:mdp}, \eg
\begin{align*}
C_k(S_k)=\min_{\{\gamma^\theta_t\}_{t=k}^{T-1}}\E\left[\sum_{t=k}^{T-1}\tilde c({S_{t+1}},\theta_t)~\Big|S_k\right]\\
=\min_{\gamma_k^\theta}\E\left[\tilde{c}(S_{k+1},\theta_k)+C_{k+1}(S_{k+1})~|S_k\right].
\end{align*}
The following theorem characterizes the structure of $C_k(s)$ for all $s\in \R^{2n}$ and $k=0,1,\ldots,T-1$.\\

\begin{thm} \label{T:optmdp}
For each $k=0,1,\ldots,T-1$, there exist a matrix $\Phi_k\in \R^{2n\times 2n}$ and $M$ functions $\{\psi^1_k(\cdot),\psi^2_k(\cdot),\ldots,\psi^M_k(\cdot)\}$, where $\psi_k^i:\R^{2n}\to \R$ for all $i=1,2,\ldots,M$, such that
\begin{align*}
C_k(S_k)=S_k\T \Phi_kS_k+\min_{i=1,\ldots,M}\{\psi^i_k(S_k)\}.
\end{align*} 
Moreover, for each $k$, the set $\R^{2n}$ can be partitioned into $M$ disjoint regions $\{\mathcal{R}^1_k,\mathcal{R}^2_k,\ldots,\mathcal{R}^M_k\}$, such that 
\begin{align*}
\theta^{*}_k=\gamma^{\theta*}_k(S_k)=[1_{\mathcal{R}^1_k}(S_k),\ldots,1_{\mathcal{R}^M_k}(S_k)]\T.
\end{align*} 
\end{thm}
\vspace*{12 pt}
\begin{proof}
The proof is provided in the Appendix.
\end{proof}

\vspace*{12 pt}
Theorem \ref{T:optmdp}  characterizes the structure of the value function associated with the quantizer-selection problem.
 However, computing the expressions for $\psi^i_k(S_k)$ is non-trivial and we need to seek an approximation. 
 Alternatively, instead of approximating $\psi^i_k(\cdot)$, one may directly approximate $C_k(\cdot)$ itself. 
In general, finding the best approximation of the value function still remains a challenging problem, and hence characterizing the best approximation for $C_k(\cdot)$ or $\psi^i_k(\cdot)$ is beyond the scope of this work.

 The optimal controller is fully characterized by the Riccati equations provided in Theorem \ref{T:optcont}, which can be computed offline. On the other hand, the optimal quantization scheme is characterized by the MDP presented in Theorem \ref{T:optmdp}.\\
 
 \begin{corr}
 Under the quantized (noise) feedback structure \eqref{E:measure}, the optimal controller is of \textit{certainty-equivalence} type. The optimal controller and optimal quantization selection problem can be decoupled and solved independently.
 \end{corr}
 
 \subsection{Quantized Measurement Quantizer Selection} \label{S:openloop}

In the previous section we analyzed the case when the quantizer selector had access to the state observation $X_t$ for all time $t$ and the optimal quantizer selection strategy $\gamma^{\theta*}_t(\cdot)$ was constructed based on the information $\I_t^q$. 
In this section we will consider the case when the decision-maker that performs the quantizer selection does not have access to the state $X_t$,
but  rather it receives the same quantized measurement \eqref{E:measure}  that the controller also receives. 
Specifically, the information set  $\I^q_t$ in this case is $\I^q_t = \{\Y_{t-1},\U_{t-1},\Theta_{t-1}\}$. 
The information of the controller remains the same as in the preceding analysis. 
 We have $\I^q_0=\varnothing$, $\I^c_t=\I^q_t\cup\{Y_t,\theta_t\}$, $\I^q_{t+1}=\I^c_t\cup \{U_t\}$. A schematic diagram showing the interaction of different components is provided in Figure \ref{Fig:schematic-open}.
 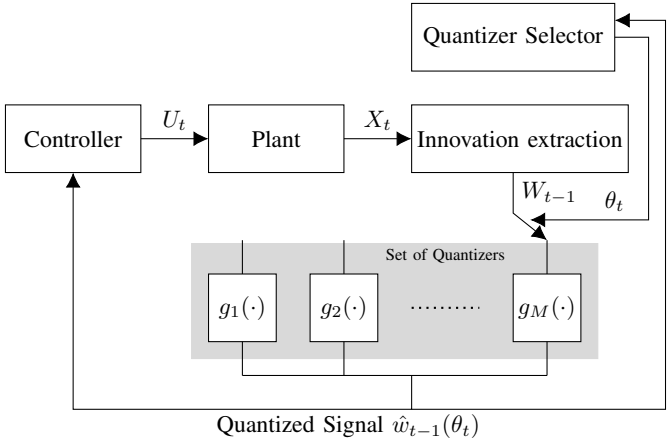
\begin{figure}
\begin{tikzpicture}[scale=0.9, every node/.style={transform shape}]
\centering
\draw[fill, color=gray!30] (2.75,-2.75) rectangle +(6,1.7);
\node[anchor=north] at (6.5,-1) {\scriptsize{Set of Quantizers}};
\draw[] (0,0) rectangle +(2,1);
\node at (1,.5) {Controller};
\draw[] (3,0) rectangle +(2,1);
\node at (4,.5) {Plant};
\draw[] (6,1.5) rectangle +(3,1);
\node at (7.5,2) {Quantizer Selector};
\draw[] (6,0) rectangle +(3.2,1);
\node at (7.6,.5) {Innovation extraction};
\draw[] (3.5,-3) --(8,-3);
\draw[fill=white] (3,-2.5) rectangle +(1,1);
\node at (3.5, -2) {$g_1(\cdot)$}; 
\draw[fill=white] (7.5,-2.5) rectangle +(1,1);
\node at (8, -2) {$g_M(\cdot)$}; 
\draw[fill=white] (4.5,-2.5) rectangle +(1,1);
\node at (5, -2) {$g_2(\cdot)$}; 
\draw[thick, dotted] (6,-2) -- (7,-2); 
\draw[->] (6,-3)-- (6, -3.5) -- (1,-3.5)-- (1,0);
\draw[->] (2,.5) --(3,.5);
\node[anchor=south] at (2.5,.5) {$U_t$};
\draw[->] (5,.5) --(6,.5);
\node[anchor=south] at (5.5,.5) {$X_t$};
\draw (3.5,-2.5)-- (3.5,-3);
\draw (5,-2.5)-- (5,-3);
\draw (8,-2.5)-- (8,-3);
\draw (3.5,-1.5)-- (3.5,-1);
\draw (5,-1.5)-- (5,-1);
\draw (8,-1.5)-- (8,-1);
\draw[->] (7.5, 0)--(7.5, -.6) -- (8,-1);
\node[anchor =west] at (7.5, -0.3) {$W_{t-1}$};
\node[anchor=west] at (3,-3.75) {Quantized Signal $\hat{w}_{t-1}(\theta_t)$};
\draw[->] (9,2)--(9.5,2)--(9.5,-.7)--(7.75,-.7);
\node[anchor=south] at (9,-.7) {$\theta_t$};
\draw[->] (6, -3.5) -- (9.75,-3.5)-- (9.75,2.25)--(9,2.25);
\end{tikzpicture}
\caption{Schematic diagram of the \textit{quantized measurement quantizer selection} system, 
where the gray block contains the set of quantizers and the desirable quantizer ($g_i$) is selected by the quantizer selector variable $\theta$.} \label{Fig:schematic-open}
\end{figure} 

In order to solve the optimization problem under this information structure, let us consider the value function 
\begin{align*}
V_k(\I^q_k)=\min_{\{\gamma^u_t\}_{t=k}^{T-1},\{\gamma^\theta_t\}_{t=k}^{T-1}}&\E\Big[\sum_{t=k}^{T-1}(X_t\T Q_1X_t+U_t\T RU_t+\theta_t\T \Lambda) \nonumber\\&+X_T\T Q_2X_T~\big|\I^q_k \Big],
\end{align*}
where it is implicitly assumed that $~U_t=\gamma^u_t(\I^c_t), \theta_t=\gamma^\theta_t(\I^q_t)$, for all $t= k,k+1,\ldots,T-1$. 
With this definition  of the value function, we have $\min_{\gamma^\U\in \Gamma^\U,\gamma^\Theta\in\Gamma^\Theta}J(\gamma^\U,\gamma^\Theta)=\E[V_0(\I^q_0)]$.
Using the fact that $\I^c_t\supseteq \I^q_t$, we may write
\begin{align*}
V_k(\I^q_t)=\min_{\{\gamma^u_t\}_{t=k}^{T-1},\{\gamma^\theta_t\}_{t=k}^{T-1}}&\E\Big[\E\Big[\sum_{t=k}^{T-1}(X_t\T Q_1X_t+U_t\T RU_t+\theta_t\T \Lambda) \nonumber\\&+X_T\T Q_2X_T~\big|\I^c_k\Big]~|~\I^q_k \Big].
\end{align*}
Using the dynamic programming principle, we can write, equivalently,
\begin{align*}
V_k(\I^q_k)=\min_{\gamma^u_k,\gamma^\theta_k}&\E\Big[\E\Big[(X_k\T Q_1X_k+U_k\T RU_k+\theta_k\T \Lambda) \nonumber\\&+V_{k+1}(\I^q_{k+1})~\big|\I^c_k\Big]~|~\I^q_k \Big].
\end{align*}
Assume that $V_k(\I^q_k)$ has the form
\begin{align} \label{E:valueopenloop}
V_k(\I^q_k)=\E[X_k\T P_kX_k+\Delta_{k-1}\T {\Pi_k}\Delta_{k-1}~|~\I^q_k]+\eta_k,
\end{align}
where $\eta_k$ and $ \Pi_k$ do not depend on the history $\U_{k-1}$ and $\Theta_{k-1}$. In \eqref{E:valueopenloop} $P_k$ is the Riccati equation given in Theorem \ref{T:optcont}, and $\Pi_k,\eta_k$ satisfy the equations
\begin{subequations}
\begin{align}
\Pi_k&=A\T (\Pi_{k+1}+N_k)A, \label{E:PI}\\
\Pi_T&={0},\\
\eta_k&=\eta_{k+1}+\tr(P_{k+1}\W) \label{E:r}\\&~~+\min_{\gamma^\theta_k}\{\tr\Big(\big(\Pi_{k+1}+N_k\big)\big(\W-F_{k-1}(\theta_k)\big)\Big)+\theta_k\T \Lambda\},\nonumber\\
\eta_T&=0,\\
N_k&=L_k\T (R+B\T P_{k+1}B)L_k, \nonumber
\end{align}
\end{subequations}
where $F_{t-1}(\theta_t) \triangleq  \text{Cov}(\hat{w}_{t-1}~|\I^q_{t})$.

Let us compute first $V_{T-1}(\I^q_{T-1})$ to obtain
\begin{align*}
V_{T-1}=&\min_{\gamma^u_{T-1},\gamma^\theta_{T-1}}\E\Big[X_{T-1}\T Q_1X_{T-1}+U_{T-1}\T RU_{T-1}\\&+(AX_{T-1}+BU_{T-1})\T Q_2(AX_{T-1}+BU_{T-1})\\&+\theta_{T-1}\T \Lambda~|\I^q_{T-1}\Big]+\tr(Q_2\W)\\
=&\min_{\gamma^u_{T-1},\gamma^\theta_{T-1}}\E\Big[X_{T-1}\T P_{T-1}X_{T-1}+\theta_{T-1}\T \Lambda\\&+\|U_{T-1}+L_{T-1}X_{T-1}\|^2_{(R+B\T Q_2B)}~|\I^q_{T-1}\Big]\\&+\tr(Q_2\W).
\end{align*}

From the above expression, the optimal $\gamma^{u*}_{T-1}(\I^c_{T-1})$ is given by 
\begin{align*}
U^*_{T-1}=\gamma^{u*}_{T-1}(\I^c_{T-1})=-L_{T-1}\E[{X}_{T-1}|\I^c_{T-1}].
\end{align*}
Thus,
\begin{align*}
V_{T-1}(\I^q_{T-1})=&\E\left[X_{T-1}\T P_{T-1}X_{T-1}~|\I^q_{T-1}\right]+\tr(Q_2\W)\\
+&\min_{\gamma^\theta_{T-1}}\E\left[\Delta_{T-1}\T N_{T-1}\Delta_{T-1}+\theta\T _{T-1}\Lambda~|\I^q_{T-1}\right].
\end{align*}

Due to the nested information structure $\I^c_{T-1}\supseteq \I^q_{T-1}$, it follows that
 $$\E[\Delta_{T-1}|\I^q_{T-1}]=\E[X_{T-1}- E[X_{T-1}|~\I^c_{T-1}]~|\I^q_{T-1}]=0.$$
Using \eqref{E:del}, we have $\Delta_{T-1}=A\Delta_{T-2}+W_{T-2}-\hat{w}_{T-2}$. It can be verified that 
\begin{align*}
&\E[\Delta_{T-1}\T N_{T-1}\Delta_{T-1}|\I^q_{T-1}]=\E[\Delta_{T-2}A\T N_{T-1}A\Delta_{T-2}|\I^q_{T-1}]\\&~~~~~~~~~~~~~+\tr(N_{T-1}\W)-\E[\hat{w}_{T-2}\T N_{T-1}\hat{w}_{T-2}|~\I^q_{T-1}].
\end{align*}
Since $\Delta_{T-2}$ does not depend on $\theta_{T-1}$, we have
\begin{align*}
V_{T-1}(\I^q_{T-1})&=\E\left[X_{T-1}\T P_{T-1}X_{T-1}~|\I^q_{T-1}\right]\\&+\tr((Q_2+N_{T-1})\W)\\
&+\E[\Delta_{T-2}A\T N_{T-1}A\Delta_{T-2}|\I^q_{T-1}]\\
&+\min_{\gamma^\theta_{T-1}}\E\left[-\hat{w}_{T-2}\T N_{T-1}\hat{w}_{T-2}+\theta_{T-1}\T \Lambda~|\I^q_{T-1}\right]\\
&=\E\left[X_{T-1}\T P_{T-1}X_{T-1}~|\I^q_{T-1}\right]+\tr(Q_2\W)\\
&+\E[\Delta_{T-2}A\T N_{T-1}A\Delta_{T-2}|\I^q_{T-1}]\\
&+\min_{\gamma^\theta_{T-1}}\{\theta_{T-1}^{\T}\Lambda+\tr(N_{T-1}\big(\W-F_{T-2}(\theta_{T-1})\big)\},
\end{align*}
where we have used the fact that\footnote{
\begin{align*} 
\E[\hat{w}_{t-1}|\I^q_{t}]&\overset{(a)}{=} \E[\hat{w}_{t-1}]=0,\\
\text{Cov}(\hat{w}_{t-1}|\I^q_{t})&=\E\left[\left(\sum_{i=1}^M\theta^i_{t}\hat w^i_{t-1}\right)\left(\sum_{i=1}^M\theta^i_{t}\hat w^i_{t-1}\right)\T|\I^q_t\right]\\
&\overset{(b)}{=}\sum_{i=1}^M\theta^i_{t}\E[\hat w^i_{t-1}\hat{w}^{i\T}_{t-1}|\I^q_t]\overset{(a)}{=}\sum_{i=1}^M\theta^i_{t}\E[\hat w^i_{t-1}\hat{ w}^{i\T}_{t-1}]
\end{align*}
where (a) follows from the fact that the random variable $\hat w^i_{t-1}$ is independent of the $\sigma$-field generated by $\I^q_t$ since the former is a function of $W_{t-1}$ and the latter is a function of $W_{-1},\ldots,W_{t-2}$, and (b) follows from the fact that $\theta_t$ is $\I^q_t$ measurable.
} 
$F_{t-1}(\theta_t) \triangleq \text{Cov}(\hat{w}_{t-1}~|\I^q_{t})=\sum_{i=1}^M\theta^i_{t}\E[\hat w^i_{t-1}\hat{ w}^{i\T}_{t-1}]  $ where 
\begin{equation*}
F_{t-1}(\theta_{t})=\sum_{i=1}^M\theta^i_tF^i_{t-1},
\end{equation*}
and
\begin{align*}
F^i_{t-1} &\triangleq \E[\hat w^i_{t-1}\hat{ w}^{i\T}_{t-1}] \\
=\sum_{i=1}^{\ell_i}&\mathsf{P}(W_{t-1}\in \p^i_j)\E[W_{t-1}|W_{t-1}\in \p^i_j]\E[W_{t-1}|W_{t-1}\in \p^i_j]\T. 
\end{align*}


Using the definitions of $\Pi_k$ and $\eta_k$,
$V_{T-1}(\I^q_{T-1})$ can be  rewritten as:
\begin{align*}
V_{T-1}(\I^q_{T-1})&=\E[\|X_{T-1}\|^2_{P_{T-1}}+\|\Delta_{T-2}\|^2_{\Pi_{T-1}}|\I^q_{T-1}]\\&+\eta_{T-1}.
\end{align*}
Note that $\eta_{T-1}$ does not depend on the history $\U_{T-2}$ and $\Theta_{T-2}$. Therefore, the hypothesis about $V_k(\I^q_k)$ is true at $k=T-1$.

Now we assume that the same is  true for $V_{k+1}(\I^q_{k+1})$. It follows that
\begin{align*}
V_k(\I^q_k)=\min_{\gamma^u_k,\gamma^\theta_k}\E\Big[&(X_k\T Q_1X_k+U_k\T RU_k+\theta_k\T \Lambda) \nonumber\\&+V_{k+1}(\I^q_{k+1})~|~\I^q_k \Big]\\
=\min_{\gamma^u_k,\gamma^\theta_k}\E\Big[&(X_t\T Q_1X_t+U_t\T RU_t+\theta_t\T \Lambda) \nonumber\\& +\|X_{k+1}\|^2_{P_{k+1}}+\|\Delta_{k}\|^2_{\Pi_{k+1}}~|~\I^q_k\Big]+\eta_{k+1}.
\end{align*}
In the above equation $\eta_{k+1}$ is taken out of the minimization since it does not depend on $\U_k$ and $\Theta_k$. Using completion of squares yields
\begin{align*}
V_k(\I^q_k)&=\min_{\gamma^u_k,\gamma^\theta_k}\E\Big[X_k\T P_kX_k+\|U_k+L_kX_k\|^2_{R+B\T P_{k+1}B}\\
&+\|\Delta_{k}\|^2_{\Pi_{k+1}}+\theta_t\T \Lambda~|~\I^q_k\Big]+\eta_{k+1}+\tr(P_{k+1}\W).
\end{align*}
Clearly, $U_k^*=\gamma^{u*}_k(\I^c_k)=-L_k\tilde{X}_k$, and hence
\begin{align*}
V_k(\I^q_k)&=\min_{\gamma^\theta_k}\E\Big[\|\Delta_{k}\|^2_{\Pi_{k+1}+N_k}+\theta_t\T \Lambda~|~\I^q_k\Big]\\
&+\E\Big[X_k\T P_kX_k~|~\I^q_k\Big]+\eta_{k+1}+\tr(P_{k+1}\W),
\end{align*}
where $N_k=L_k\T (R+B\T P_{k+1}B)L_k$.
Using the fact
\begin{align*}
\E\Big[\|\Delta_{k}\|^2_{\Pi_{k+1}+N_k}~|~\I^q_k\Big]&=\E\Big[\|\Delta_{k-1}\|^2_{A\T (\Pi_{k+1}+N_k)A}~|~\I^q_k\Big]\\
&+\tr((\Pi_{k+1}+N_k)(\W-F_{k-1}(\theta_k))),
\end{align*}
we obtain
\begin{align*}
V_k(\I^q_k)&=\min_{\gamma^\theta_k}\{\tr\Big(\big(\Pi_{k+1}+N_k\big)\big(\W-F_{k-1}(\theta_k)\big)\Big)+\theta_k\T \Lambda\}\\
&+\E\Big[X_k\T P_kX_k+\|\Delta_{k-1}\|^2_{A\T (\Pi_{k+1}+N_k)A}~|~\I^q_k\Big]\\&+\eta_{k+1}+\tr(P_{k+1}\W),
\end{align*}
and using \eqref{E:PI} and \eqref{E:r}, we obtain
\begin{align*}
V_k(\I^q_k)&=\E[X_k\T P_kX_k+\Delta_{k-1}\T {\Pi_k}\Delta_{k-1}~|~\I^q_k]+\eta_k.
\end{align*}
Under the \textit{quantized measurement information pattern} the controller retains the same structure as for the \textit{perfect measurement information pattern}. 
However, the quantization selection strategy has changed, as expected. 
In this case, the optimal quantizer selection strategy is given by
\begin{align} \label{E:thta}
\gamma^{\theta*}_k(\I^q_k)=\arg\min_{\theta\in \A}\{&\tr\big(\Omega_kM_{k-1}(\theta)\big)+\theta\T \Lambda\}.
\end{align}
where $M_{k-1}(\theta)=\W-F_{k-1}(\theta)$ and $\Omega_k=\Pi_{k+1}+N_k$ for all $k=0,1,\ldots,T-1$.
The optimal selection of $\theta_k^*$ does neither depend on the previous choices ($\theta_1^*,\ldots,\theta^*_{k-1}$) nor depends on the future  choices $\theta_k^*,\ldots,\theta_{T-1}^*$.
The optimal quantizer at time $k$ is the one which reduces the weighted noise estimation error covariance ($M_k(\theta)$) most with the least cost ($\lambda_i$).

\begin{pr}
For all $k=0,1,\ldots,T-1$, $\Omega_k=\Upsilon_k-P_k$, where $\Upsilon_k$ satisfies the dynamics
\begin{align} \label{E:upsilon}
\Upsilon_k&=A\T\Upsilon_{k+1}A+Q_1, \\
\Upsilon_T&=Q_2. \nonumber
\end{align}

\end{pr}
\begin{proof}
The proposition is proved by showing that $\Upsilon_k-P_k=\Pi_{k+1}+N_k$ for all $k=0,1,\ldots,T-1$. At $k=T-1$,
\begin{align*}
\Upsilon_{T-1}=&A\T Q_2A+Q_1\\
=&A\T Q_2 A+Q_1-N_{T-1}+N_{T-1}\\
=&P_{T-1}+N_{T-1}.
\end{align*}
Therefore, $\Upsilon_{T-1}-P_{T-1}=\Pi_T+N_{T-1}$ since $\Pi_T=0$. Thus, the relationship holds for $k=T-1$. We shall use mathematical induction to prove the proposition. Let us assume that the relationship holds for some $k=1,\ldots,T-1$, then
\begin{align*}
\Upsilon_{k-1}-P_{k-1}=&A\T \Upsilon_k\A+Q_1-(Q_1+A\T P_kA-N_{k-1})\\
=&A\T (\Upsilon_k-P_k)A+N_{k-1}.
\end{align*}
Using the hypothesis that $\Upsilon_k-P_k=\Pi_{k+1}+N_k$ and \eqref{E:PI}, we obtain
\begin{align*}
\Upsilon_{k-1}-P_{k-1}=&A\T (\Pi_{k+1}+N_k)A+N_{k-1}\\
=&\Pi_k +N_{k-1}.
\end{align*}
Therefore, for all $k=0,1,\ldots,T-1$, $\Omega_k=\Upsilon_k-P_k$.
\end{proof}

The purpose of the quantizer in this case is to reduce the variance of the noise at the controller. 
The corresponding reduction in the noise covariance by selecting the $i$-th quantizer is $\W-F_{k-1}^i$. 
Equation \eqref{E:thta} also shows that the importance of the noise covariance reduction at different time instances is different. 
The weights $\Omega_k=\Pi_{k+1}+N_k$ at each time $k$ denote the expected effect that the noise would have for the remaining horizon $t=k,k+1,\ldots,T-1$.

We summarize the results for the quantized measurement optimal quantizer selection in the following theorem.\\

\begin{thm}
Given the information patterns $\I^q_k=\{\Y_{k-1},\U_{k-1},\Theta_{k-1}\}$ and $\I^c_k=\{\Y_k,\U_{k-1},\Theta_k\}$, the optimal control policy $\gamma^{u*}_k:\I^c_k\to \R^m$ is given by
\begin{align*}
U^*_k=\gamma^{u*}_k=-L_k\E\left[X_k~|~\I^c_k\right],
\end{align*}
and the optimal quantizer selection policy is given by
\begin{align*}
\theta^*_k=\gamma^{\theta*}_k(\I^q_k)=\arg\min_{\theta\in \A}\{&\tr\big(\Omega_kM_{k-1}(\theta)\big)+\theta\T \Lambda\},
\end{align*}
where $M_{k-1}(\theta)=\W-F_{k-1}(\theta)$ and $\Omega_k=\Upsilon_k-P_k$, $\Upsilon_k$ satisfies the dynamics \eqref{E:upsilon}.
\end{thm}

\begin{proof}
The proof of this theorem follows from the construction of the value function $V_k(\I^q_k)$ followed by the analysis presented in this section to show that the $\gamma^{u*}_k$ and $\gamma^{\theta*}_k$ described in the the theorem are the optimal strategies for the value function.
\end{proof}

\vspace*{12 pt}
From the expression of $\gamma^{\theta*}_k(\I^q_k)$, we notice that $\gamma^{\theta*}_k(\cdot)$ is only a function of $k$, since the computation of any of the parameters $\Omega_k,F_{k-1}$ does not require  knowledge of $\I^q_k$. Thus, one may compute $\gamma^{\theta^*}_k(\cdot)$ without having access to $\I^q_k$. Moreover, $\Pi_k,N_k,F_k$ can be computed offline and hence $\gamma^{\theta*}_k(\cdot)$ can be computed offline too. 

{Compared to the results obtained in Section \ref{S:closeloop}, the optimal quantizer selection problem under the \textit{quantized measurement information structure} does not require the solution of an MDP. 
The optimal quantizer selection strategy presented in \eqref{E:thta} is a sub-optimal strategy for the MDP \eqref{E:mdp}. 
However, this sub-optimal strategy can be used as an initial guess to solve \eqref{E:mdp} via iterative techniques such as policy iteration. }

\begin{rem}[Constrained optimization]
In the previous sections we have considered a problem  of the form:
\begin{align*}
\min_{\gamma^\U,\gamma^\Theta}\{J_{\rm{LQG}}+J_{\rm{quant}}\},
\end{align*}
However, a possibly more interesting problem is to consider a constrained optimization problem of the form
\begin{align} \label{E:cons1}
\min_{\gamma^\U,\gamma^\Theta} ~&J_{\rm{LQG}},\\ \nonumber
\text{s.t.~} &J_{\rm{quant}} \le \mathcal{B},
\end{align}
Although solving the optimization problem \eqref{E:cons1} is beyond the scope of this paper, however, the solution of this problem can be constructed by solving the family of optimization problems:
\begin{align} \label{E:weighted}
\min_{\gamma^\U,\gamma^\Theta}\{\beta J_{\rm{LQG}}+(1-\beta)J_{\rm{quant}}\},
\end{align}
for all $\beta\in [0,1]$.
\eqref{E:weighted} can be solved by the framework presented in this paper.
\end{rem}

%
%
%
%
%

\section{SIMULATION RESULTS} \label{S:simu}
\subsection{Example1: Unstable System}
Let us consider the two-dimensional (unstable) system
\begin{align*}
X_{t+1}=\begin{bmatrix}
1.01 &0.5\\ 0 &1.1
\end{bmatrix}X_t+\begin{bmatrix}
0.1 &0\\ 0 &0.15
\end{bmatrix}U_t+W_t,
\end{align*}
with initial condition $X_0~\sim\N(0,I)$ and $W_t\sim \N(0,\frac{1}{4}I)$. 
The control cost has  parameters $Q=Q_f=R=\frac{1}{2}I$. 
The time horizon was set to $T=50$.
 The simulation was performed with a  scenario of three quantizers ($\Q^1,\Q^2,\Q^3$) where $\Q^i$ has $2^i$ number of quantization levels. 
 The partitions associated with the quantizers are $\p^1=\{\R_+\times\R, \R_{< 0}\times \R\}$, $\p^2=\{\R_+\times \R_+,~\R_+\times\R_{< 0},~ \R_{<0}\times\R_+, \R_{<0}\times\R_{<0}\}$ and $\p^3=\{[0,0.5)\times\R_+, ~[0.5,\infty)\times\R_+,[0,0.5)\times\R_{<0}, ~[0.5,\infty)\times\R_{<0},[-0.5,0)\times\R_+, ~(-\infty,-0.5)\times\R_+,[-0.5,0)\times\R_{<0}, ~(-\infty,-0.5)\times\R_{<0}\}$. The costs associated with the quantizers are $\Lambda=[1 , 2, 3]\T \times 10^4$.
For this example, we consider the \textit{quantized measurement information pattern} for the quantizer selection. Therefore, the optimal quantizer at time $t$ is selected based on equation \eqref{E:thta}. The optimal selection of the quantizers are plotted in Figure \ref{F:optimal quantizer}.
\begin{center}
\begin{figure}
\includegraphics[scale=.5]{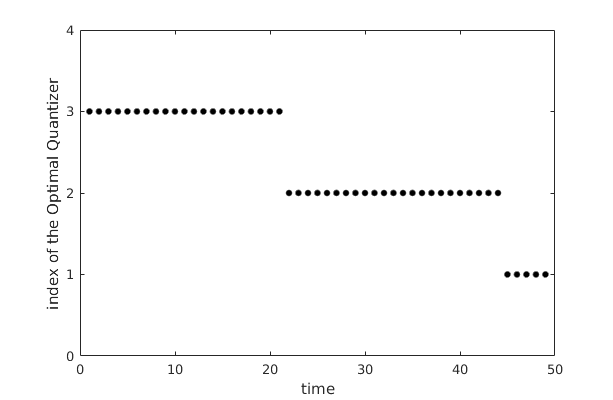}
\caption{Optimal quantizer selection over time.} \label{F:optimal quantizer}
\end{figure}
\end{center}
To characterize the utilization quotient of each quantizer, we define the variable $\rho_i(t)$ as
\begin{align*}
\rho_i(t)=\frac{\text{\# utilization of }i\text{-th quantizer up to time }t }{t}.
\end{align*}
Note that $\sum_{i=1}^3\rho_i(t)=1$ for all $t$. The optimal utilization of the quantizers are plotted in Figure \ref{F:utilization}.
\begin{figure}
\centering
\includegraphics[scale=0.75]{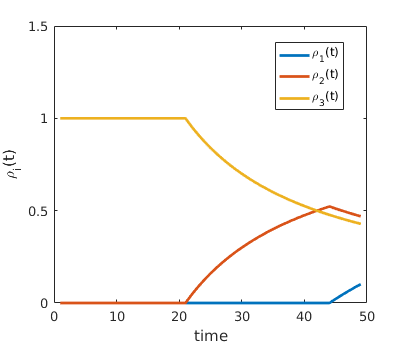}
\caption{Utilizations of different quantizers over time.} \label{F:utilization}
\end{figure}

The Pareto curve for the bi-objective optimization problem is shown in Figure \ref{F:Pareto}. As can be seen from Figure \ref{F:Pareto}, the minimum control cost  achievable for this problem is $J_{\rm{LQG}}=2.295\times 10^6$ with the maximum quantization cost $J_{\rm{quant}}=1.5\times 10^6$. On the other hand, the maximum control cost is $3.367\times 10^6$ when the quantization cost is kept at a minimum ($5 \times 10^5$).
One interesting observation for this particular problem is the steepness of the Pareto curve. 
In this study it shows that the quantization cost can be reduced drastically with very minor change in the control cost. However, after a point, a slight the reduction in the quantization cost leads to large change in the control cost.
Thus, it appears that the rate of reduction in control cost by changing the number of transmission bits is non-uniform, and furthermore, beyond a certain quantization data-rate the rate of change in control cost is negligible.
In a future study it would be interesting to study whether such behavior is fundamental to the LQG problem or only occurs when certain choices of parameters ($A,B,\Q$ etc.) are taken.
\begin{figure}
\centering
\includegraphics[scale=0.5]{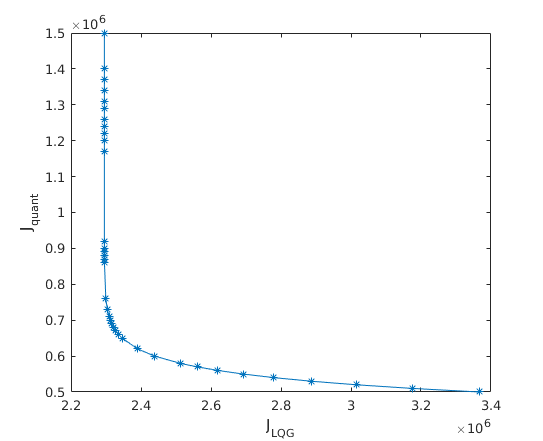}
\caption{Pareto front of the bi-objective problem.} \label{F:Pareto}
\end{figure}

\subsection{Example2: Stable System}
We also performed a similar experiment with the stable dynamics:
\begin{align*}
X_{t+1}=\begin{bmatrix}
0.9 &0.2\\ 0 &0.7
\end{bmatrix}X_t+\begin{bmatrix}
0.1 &0\\ 0 &0.15
\end{bmatrix}U_t+W_t,
\end{align*}
where the initial condition $X_0~\sim\N(0,I)$ and $W_t\sim \N(0,\frac{1}{4}I)$. All the other parameters were chosen to be exactly the same as in Example 1 except the fact that the quantizer costs are now $\Lambda=[0.03, 0.06, 0.09]\T $. In this case, we observed a similar quantizer utilization pattern (Figure \ref{F:utilization_stable}) as we observed for the unstable dynamics case. 
\begin{figure}
\centering
\includegraphics[scale=0.65]{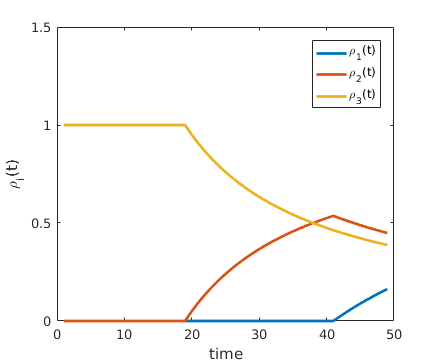}
\caption{Utilizations of different quantizers over time.} \label{F:utilization_stable}
\end{figure}
When the quantizer cost were kept to $\Lambda=[1,2,3]\T $, the optimal quantizer at any time was the 1st quantizer (the cheapest one). 
In this case, the increment in quantization cost for using a better resolution quantizer (at any time) is higher than the decrement in the control cost using the better quantizer. Thus, the quantization cost seems to be \textit{too high} for such a stable system.

Also, when the quantization cost was set to be $10^{-4}\times[1,2,3]$, the optimal quantizer choice was the 3rd quantizer (the best resolution quantizer). In this case, the decrement in the quantization cost for using a lower resolution quantizer is smaller than the increment in the control performance. 

For this example we also implement the perfect state-feedback strategy, i.e., the controllers had a perfect measurement as opposed to quantized measurements. The deviation in the state trajectories and control inputs between the optimal perfect feedback scenario and optimal quantized are shown in Figures \ref{F:state_trajectory} and \ref{F:Control_trajectory}. Combining Figures \ref{F:utilization_stable}, \ref{F:state_trajectory} and \ref{F:Control_trajectory} we notice that by optimally using the quantizers, the control and state trajectories are able to mimic the perfect-feedback trajectory very closely. Since the quantizers considered in this example have 2, 4 and 8 quantization levels, the number of bits required to transmit the measurements will be 1, 2 and 3 bits respectively. For this example, the average bit rate (bits per sample) is $(1/T) \sum_{i=0}^{T-1}\theta^i_t\log_2\ell_i =2.22$.

\begin{figure}
\centering
\includegraphics[scale=0.6]{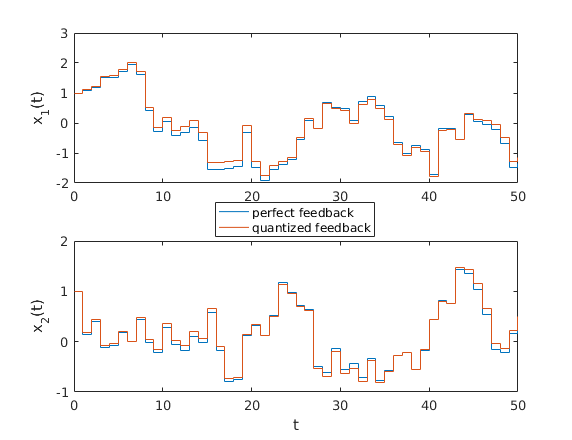}
\caption{Top: First component of the state under perfect feedback (blue) and quantized feedback (red). Down: Second component of the state under perfect feedback (blue) and quantized feedback (red).} \label{F:state_trajectory}
\end{figure}
\begin{figure}
\centering
\includegraphics[scale=0.6]{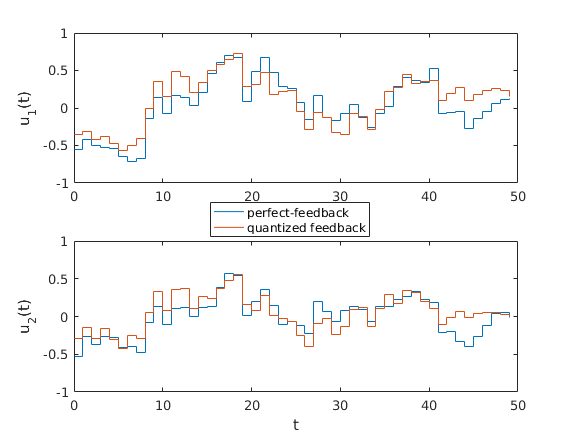}
\caption{Top: First component of control input under perfect feedback (blue) and quantized feedback (red). Down: Second component of control input under perfect feedback (blue) and quantized feedback (red).} \label{F:Control_trajectory}
\end{figure}

\section{Discussion and Extensions} \label{S:Discussion}

\subsection{Partial Noisy Observations}

In this work we have assumed that perfect state measurements are available at the sensor prior to quantization and the communication channel is error-free and not susceptible delay and distortion. 
The assumption on the availability of perfect state measurement can be dropped and the proposed framework can readily be extended to partially observed noise corrupted sensory measurements.
In order to incorporate partial observation where process noise $W_t$ is not readily computable one needs to construct a different signal from the noisy observations and quantize that specific signal for communication. 
A detailed study on partially observable systems can be found in our recent work \cite{maity2019optimal}.

\subsection{Delay and Distortion}

Two major aspects in considering a digital communication network is delay in the transmission and distortion. 
In this paper, we have not considered either of these aspects and rather considered an ideal network with zero delay and no distortion in the channel. 
Although we have not considered the effects of delay explicitly, this framework can easily be extended to the scenario where delay is present. 
In the presence of network delay, the information available at the controller will be affected since some of the measurements arrival will be delayed and hence the state estimation will be affected.
The detailed discussion on the effects of delay is beyond the scope of this work, and the interested readers are directed to our follow-up work in \cite{maity2019optimal}.

In order of discuss the effects on channel distortion, let us assume that the channel input alphabet is $\mathcal{I}=\cup_{i=1}^M\Q^i$  where channel output alphabet is $\mathcal{O}$,\footnote{
For simplicity, one may assume that $\mathcal{O}=\mathcal{I}$
} i.e., the channel accepts the input $q^i_j$ from the quantizers and maps it to one of the outputs $o\in \mathcal{O}$ based on the probability distribution $p_c(o|q^i_j)$.  
$p_c(\cdot|\cdot)$ is the channel characteristic which is known and hence, upon receiving an alphabet $o \in \mathcal{O}$, we can compute the posterior distribution $p(q^i_j|o)$ by Baye's rule. 
Thus, if at time $t$, the noise $W_{t-1}$ is quantized and the controller receives $o_t \in \mathcal{O}$, then 
\begin{align*}
\E[W_{t-1}~|~o_t]&=\E\left[\E[W_{t-1}~|~q^i_j,o_t]~|o_t\right]\\
&=\sum_{i=1}^M\sum_{j=1}^{\ell_i}\E[W_{t-1}~|~q_j^i]p(q^i_j~|~o_t).
\end{align*}

In the case of no distortion, we trivially have $o_t=q^i_k$ if the input to the channel at time $t$ was $q^i_k$, and hence $\E[W_{t-1}~|~o_t]=\E[W_{t-1}|g_i(W_{t-1})]$ as given in \eqref{E:indicatorw}.

\subsection{Choice of Quantizers}

The aim of this paper is to select the best quantizer from a given set of quantizers.
It is assumed that such a set of quantizers is designed a priori.
An interesting research topic is the design of such a set of possible quantizers. 
From the analysis in this paper, one can immediately see that the cost function is reduced to the form of 
\begin{align} \label{E:Qparameters}
J(\U^*,\Theta^*)=\E[X_0\T P_0X_0]+r_0+\E\left[\sum_{t=0}^{T-1}{\Delta_t^*}\T N_t\Delta_t^*+{\theta_t^*}\T \Lambda\right]
\end{align}
where the distribution of $\Delta^*_t$ (and hence $\E[\Delta^*_t{\Delta^*_t}\T]$) depends on the quantizer parameters. 
Thus, one can re-write~\eqref{E:Qparameters} as an explicit function of these parameters, 
and may further optimize $J(\U^*,\Theta^*)$ with respect to these parameters to find the best set of quantizers.
Although, at least in principle,  such a method to directly design the quantization set is possible,  
it can be computationally expensive.

\section{CONCLUSIONS} \label{S:conclusion}

In this work, we have considered a classical quantization-based LQG problem with a positive  quantization cost. 
The problem is to choose an optimal quantizer among a set of available quantizers that minimizes the combined cost of quantization and control performance. 
We have shown that the optimal controller exhibits a separation principle and it has a linear relationship with the estimate of the state. 
The optimal gains are found by solving the classical Riccati equation associated with the LQG problem. 
We have also shown that the optimal selection for the quantizers can be found by solving an auxiliary MDP problem that can be solved independently from the controller synthesis problem. 
A closed-form solution to the MDP problem is not available at this point.
 Instead, we provided some structural properties of the value function associated with the MDP. 
 These structural properties can be exploited for the purpose of value function estimation. 

We also considered a version of the problem, where the quantization selection is done based on limited information -- which we refer to as the \textit{quantized measurement information structure}. 
We showed that, under such an information pattern, the controller structure remains unchanged, and the optimal quantization selection can be solved offline.





\section*{APPENDIX}

\begin{lm} \label{L:1}
Let $X: \Omega  \to \R^{n_x}$ and $W: \Omega \to \R^{n_w}$ be two independent random variables defined over the probability space $(\Omega, \mathcal{F},\mathsf P)$, and let $f:\R^{n_x}\to \R^{n_w}$ be a measurable function. Let $\{\psi^1(\cdot),\ldots,\psi^M(\cdot) \}$ be $M$ measurable functions from $\R^{n_w}$ to $\R$. Then there exists a function $\tilde \psi: \R^{n_w}\to \R\cup\{-\infty,+\infty\}$ such that
\begin{align*}
\E\left[\min_{i}\{\psi^i(f(X)+W)\}~\big|X\right]=\tilde{\psi}(f(X)).
\end{align*}
\end{lm}

\begin{proof}
Let us denote  
\begin{align*}
R^\Psi_i=\{w\in \R^{n_w}~|~\psi^i(w)\le \psi^j(w),  j=1,2,\ldots,M\},
\end{align*}
so that 
\begin{align*}
\cup_{i=1}^MR^\Psi_i=\R^{n_w}.
\end{align*}
Note that $R^{\Psi}_i$ and $R^\Psi_j$ may not necessarily be disjoint for all $i,j$.
For all $y\in R^{n_w}$, let us  define  $R^\Psi_i(y)=\{w\in \R^{n_w}~|~\psi^i(y+w)\le \psi^j(y+w),  j=1,2,\ldots,M\}$. 
One may think of $R^\Psi_i(y)$ as the translation of the set $R^\Psi_i$ by the vector $y$. 
It is also true that $\cup_{i=1}^MR^\Psi_i(y)=\R^{n_w}$ for all $y\in \R^{n_w}$. 
If $R^\Psi_i(y)\cap R^\Psi_j(y)\ne \varnothing$ for some $y$, then we randomly assign the elements of the set $R^\Psi_i(y)\cap R^\Psi_j(y)$ either to $R^\Psi_i(y)$ or to $R^\Psi_j(y)$. 
Let us denote these sets to be $\{\tilde{R}^\Psi_i(y)\}_{i=1}^M$ after this random assignment. The random assignment is to ensure $\tilde R^\Psi_i(y)\cap \tilde R^\Psi_j(y)= \varnothing$. However, we still have $\cup_{i=1}^M\tilde{R}^\Psi_i(y)=\R^{n_w}$ and
\begin{align*}
\tilde R^\Psi_i(y)=\{w\in \R^{n_w}~|~&\psi^i(y+w)\le \psi^j(y+w), \\& j=1,2,\ldots,M\}.
\end{align*}

Using $\{\tilde{R}^\Psi_i(y)\}_{i=1}^M$ as a partition of $\R^{n_w}$, we can write
\begin{align*}
\E\left[\min_{i}\{\psi^i(f(X)+W)\}~\big|X\right]=\int_{\R^{n_w}}\min_{i}\{\psi^i(y+w)\}d\mathsf P(w),
\end{align*}
where $y=f(X)$ for notational convenience. Therefore,
 \begin{align*}
\E\left[\min_{i}\{\psi^i(f(X)+W)\}~\big|X\right]=&\sum_{i=1}^M\int_{\tilde R^{\Psi}_i(y)}\psi^i(y+w)d\mathsf P(w)\\
\triangleq&\sum_{i=1}^M\tilde{\psi}_i(y),
\end{align*}
where we define $\tilde{\psi}_i(y)=\int_{\tilde R^{\Psi}_i(y)}\psi^i(y+w)d\mathsf P(w)$. If we define $\tilde \psi(y)=\sum_{i=1}^M\tilde{\psi}_i(y)$ for all $y\in \R^{n_w}$, then 
\begin{align*}
\E\left[\min_{i}\{\psi^i(f(X)+W)\}~\big|X\right]=\tilde{\psi}(y)=\tilde{\psi}(f(X)).
\end{align*}
\end{proof}
Under the hypothesis that $X$ and $W$ are independent, and $\{\psi^i(\cdot)\}$ being a set of measurable functions from $\R^{n_w}$ to $\R$, there always exists a measurable function $\hat \psi:\R^{n_x}\to \R$ such that 
\begin{align*}
\E\left[\min_{i}\{\psi^i(f(X)+W)\}~\big|X\right]=\hat \psi(X).
\end{align*}
The above lemma  shows that $\hat \psi$ can be represented as a composition $\hat \psi(\cdot)=\tilde{\psi}\circ f(\cdot)$.

\subsection{Proof of Theorem \ref{T:optmdp}} \label{A:proof}
\begin{proof}
Let us note that at $k=T-1$ we have
\begin{align*}
C_{T-1}(s)=&\min_{\gamma^\theta_{T-1}}\E\left[\tilde{c}(S_T,\theta_{T-1}) ~|~S_{T-1}=s\right]\\
=&\min_{\gamma^\theta_{T-1}}\E\left[S_T\T \tilde N_{T-1} S_T+\theta_{T-1}\T \Lambda~|S_{T-1}=s\right]\\
=&\min_{\gamma^\theta_{T-1}}\E\Big[\Big\|HS_{T-1}-\begin{bmatrix}
G_{T-1}(S_{T-1})\theta_{T-1}\\0
\end{bmatrix}\Big\|^2_{\tilde{N}_{T-1}}\\
&~~~~+\theta_{T-1}\T \Lambda~|S_{T-1}=s\Big],
\end{align*}
where the last equality follows from $\tilde{N}_{T-1}\begin{bmatrix}
0\\W_{T-1}
\end{bmatrix}=\begin{bmatrix}
0\\0
\end{bmatrix} $.
Recall that $H=\begin{bmatrix}
A~~&{I}\\{0}~~&{0}
\end{bmatrix}$ and $G_{t}(S_t)=[\hat w_{t-1}^1(S_t), \hat w_{t-1}^2(S_t),\ldots,\hat w_{t-1}^M(S_t)]$ where $\hat w_{t-1}^i(S_t)=\E[W_{t-1}~|~g_i(W_{t-1})]$, 
 and $S_t=[\Delta_{t-1}\T , W_{t-1}\T ]\T $. Here we write the argument in $\hat w^i_{t-1}(\cdot)$ to emphasize the fact that it depends on $S_t$.  

One can verify that 
\begin{align*}
\Big\|&\begin{bmatrix}
G_{T-1}(S_{T-1})\theta_{T-1}\\0
\end{bmatrix}\Big\|^2_{\tilde{N}_{T-1}}\\&=\theta_{T-1}\T G_{T-1}(S_{T-1})\T N_{T-1}G_{T-1}(S_{T-1})\theta_{T-1}\\&=\sum_{i=1}^M\big(\hat{w}^i_{T-2}(S_{T-1})\big)\T N_{T-1}\big(\hat{w}^i_{T-2}(S_{T-1})\big)\theta_{T-1}^i\\
&\triangleq\tilde{G}_{T-1}(S_{T-1})\theta_{T-1},
\end{align*}
where $\tilde{G}_{T-1}(S_{T-1})$ is the vector $[\|\hat w^1_{T-2}(S_{T-1})\|_{N_{T-1}}^2, \ldots ,\|\hat w^M_{T-2}(S_{T-1})\|_{N_{T-1}}^2]$.

Thus, after some simplifications,
\begin{align*}
C_{T-1}(s)=&\min_{\gamma^\theta_{T-1}}\Big\{S_{T-1}\T H\T \tilde{N}_{T-1}HS_{T-1}\\&-2S_{T-1}\T \begin{bmatrix}
A\T {N}_{T-1}G_{T-1}(S_{T-1})\\{0}
\end{bmatrix}\theta_{T-1}\\
&+\tilde{G}_{T-1}(S_{T-1})\theta_{T-1}+\theta_{T-1}\T \Lambda\Big\}\\
=&\min_{\gamma^\theta_{T-1}}\Big\{ \Psi_{T-1}(S_{T-1})\theta_{T-1}\Big\}\\&+S_{T-1}\T H\T \tilde{N}_{T-1}HS_{T-1}
\end{align*}
where $\Psi_{T-1}(S_{T-1})=-2S_{T-1}\T \begin{bmatrix}
A\T {N}_{T-1}G_{T-1}(S_{T-1})\\{0}
\end{bmatrix}+\tilde{G}_{T-1}(S_{T-1})+\Lambda\T $ is an $M$ dimensional row vector. 
Let us denote by $\Psi_{T-1}(\cdot)=[\psi^i_{T-1}(\cdot),\psi^2_{T-1}(\cdot),\ldots,\psi^M_{T-1}(\cdot)]$, where $\psi^i_{T-1}:\R^{2n}\to \R$ is the $i$-th component of $\Psi_{T-1}$.
Thus, $\min_{\gamma^\theta_{T-1}}\Big\{ \Psi_{T-1}(S_{T-1})\theta_{T-1}\Big\}=\min_i\{\psi^i_{T-1}(S_{T-1})\}$. Let
\begin{align*}
i^*(S_{T-1})=\arg\min_i\{\psi^i_{T-1}(S_{T-1})\}.
\end{align*}
If $i^*(S_{T-1})$ is not unique, then without loss of generality,  one of the minimizers is randomly selected.\footnote{In case there are multiple $i \in \{0,1,\ldots,M\}$ that minimizes $\min_i\{\psi^i_{T-1}(S_{T-1})\}$, let us introduce the set $I^*=\{i^*_1,i^*_2,\ldots,i^*_l\}$ where  $l\le M$, $i^*_j\in \{1,2,\ldots,M\}$ and each $i^*_j$ is a minimizer. In such a case, one can use a randomized policy over this set of minimizers to select the quantizers, \eg $\gamma^\theta_{T-1}(S_{T-1})=b_j$ with probability $ p_j \in [0,1]$ where $j\in I^*$ and $\sum_{j\in I^*}p_j=1$. Nonetheless, the value $C_{T-1}(S_{T-1})$ remains unaffected by the choice of $p_j$.} 
It follows
\begin{align*}
\theta^*_{T-1}=\gamma^{\theta*}_{T-1}(S_{T-1})=b_{i^*(S_{T-1})},
\end{align*}
where $b_i\in \A$ is the $i$-th basis vector in $\R^M$.
Thus,
\begin{align*}
C_{T-1}(S_{T-1})=S_{T-1}\T \Phi_{T-1}S_{T-1}+\min_i\{\psi^i_{T-1}(S_{T-1})\},
\end{align*}
where $\Phi_{T-1}=H\T \tilde{N}_{T-1}H$. 

Let us now assume that for some $k$ it is true that
\begin{align} \label{E:CK_assumptionMDP}
C_k(S_k)=S_k\T \Phi_kS_k+\min_{i=1,\ldots,M}\{\psi^i_k(S_k)\}.
\end{align}

Recall that,
\begin{align*}
C_{k-1}(S_{k-1})=\min_{\gamma^\theta_{k-1}}\E\Big[S_k\T \tilde N_kS_k+\theta_{k-1}\T \Lambda+C_k(S_k)~|S_{k-1}\Big].
\end{align*}

Using the hypothesis, let us replace $C_k(S_k)$ in the above equation, and consequently
\begin{align} \label{E:costk}
C_{k-1}(S_{k-1})=&\min_{\gamma^\theta_{k-1}}\E\Big[S_k\T (\tilde N_k+\Phi_k)S_k+\theta_{k-1}\T \Lambda \nonumber\\&+\min_{i=1,\ldots,M}\{\psi^i_k(S_k)\}~|~S_{k-1}\Big].
\end{align}

Using the dynamics of $S_t$ in \eqref{E:s} one can write $S_{k}=f_0(k-1,S_{k-1},\theta_{k-1})+\begin{bmatrix}
0\\W_{k-1}
\end{bmatrix}$. The expression for $f_0(\cdot,\cdot,\cdot)$ can be obtained from  \eqref{E:s}
\begin{align*}
f_0(t,S,\theta)=HS-\begin{bmatrix}
G_{t}(S)\theta\\0
\end{bmatrix}.
\end{align*}
Since $S_{k-1}$ and $W_{k-1}$ are independent random variables and $\theta_{k-1}=\gamma^\theta_{k-1}(S_{k-1})$, we can use Lemma \ref{L:1}\footnote{Use $S_{k-1}$ as $X$ and $f_0(K-1,S_{k-1},\gamma^\theta_{k-1}(S_{k-1}))$ as the $f(X)$ in Lemma \ref{L:1}.} to conclude that there exists a function $\tilde \psi_k:\R^{2n}\to \R$ such that
\begin{align}\label{E:psi_tilde}
\tilde \psi_k(f_0(k-1,S_{k-1},\theta_{k-1}))=\E\left[\min_{i=1,\ldots,M}\{\psi^i_k(S_k)\}~|~S_{k-1} \right].
\end{align}

Since $\theta_{k-1}\in \A=\{b_1,b_2,\ldots,b_M\}$, we define
\begin{align*}
\tilde{\psi}^i_k(S_{k-1})\triangleq \tilde{\psi}_k(f_0(k-1,S_{k-1},b_i)).
\end{align*}
Therefore,
\begin{align} \label{E:psi_tilde_sum}
\tilde \psi_k(f_0(k-1,S_{k-1},\theta_{k-1}))=\sum_{i=1}^M\tilde{\psi}^i_k(S_{k-1})\theta^i_{k-1}.
\end{align}

Following similar steps as before, it can be shown that
\begin{align*}
\E&\left[S_k\T (\tilde N_k+\Phi_k)S_k~|S_{k-1}\right]=S_{k-1}\T H\T (\tilde N_k+\Phi_k)HS_{k-1}\\
&+\E\left[ \begin{bmatrix}
0\\W_{k-1}
\end{bmatrix}\T (\tilde N_k+\Phi_k)\begin{bmatrix}
0\\W_{k-1}
\end{bmatrix}\right]\\
&-2\sum_{i=1}^MS_{k-1}\T H\T (\tilde N_k+\Phi_k)\begin{bmatrix}
\hat w^i_{k-2}(S_{k-1})\\0
\end{bmatrix}\theta^i_{k-1}\\
&+\sum_{i=1}^M \begin{bmatrix}
\hat w^i_{k-2}(S_{k-1})\\0
\end{bmatrix}\T  (\tilde N_k+\Phi_k) \begin{bmatrix}
\hat w^i_{k-2}(S_{k-1})\\0
\end{bmatrix}\theta^i_{k-1}.
\end{align*}
Using the fact $\E[W_{k-1}W_{k-1}\T ]=\W$, we obtain
\begin{align} \label{E:sk_sim}
\E&\left[S_k\T (\tilde N_k+\Phi_k)S_k~|S_{k-1}\right]=S_{k-1}\T H\T (\tilde N_k+\Phi_k)HS_{k-1}\nonumber \\
&+\tr\Big((\tilde N_k+\Phi_k)\begin{bmatrix}
{0} & {0}  \\ {0} & \W
\end{bmatrix}\Big) \nonumber\\
&-2\sum_{i=1}^MS_{k-1}\T H\T (\tilde N_k+\Phi_k)\begin{bmatrix}
\hat w^i_{k-2}(S_{k-1})\\0
\end{bmatrix}\theta^i_{k-1} \nonumber \\
&+\sum_{i=1}^M \begin{bmatrix}
\hat w^i_{k-2}(S_{k-1})\\0
\end{bmatrix}\T  (\tilde N_k+\Phi_k) \begin{bmatrix}
\hat w^i_{k-2}(S_{k-1})\\0
\end{bmatrix}\theta^i_{k-1}.
\end{align}

Let us now define $\psi^i_{k-1}(S_{k-1})$ as follows
\begin{align} \label{E:psi_k-1}
\psi^i_{k-1}(S_{k-1})&=\tilde{\psi}^i_k(S_{k-1})+\lambda_i+\tr\Big((\tilde N_k+\Phi_k)\begin{bmatrix}
{0} & {0}\\{0} & \W
\end{bmatrix}\Big) \nonumber \\&+\begin{bmatrix}
\hat w^i_{k-2}(S_{k-1})\\0
\end{bmatrix}\T  (\tilde N_k+\Phi_k) \begin{bmatrix}
\hat w^i_{k-2}(S_{k-1})\\0
\end{bmatrix} \nonumber\\
&-2S_{k-1}\T H\T (\tilde N_k+\Phi_k)\begin{bmatrix}
\hat w^i_{k-2}(S_{k-1})\\0
\end{bmatrix}.
\end{align}
 
 Using \eqref{E:psi_tilde}, \eqref{E:psi_tilde_sum}, \eqref{E:sk_sim} and \eqref{E:psi_k-1} we can rewrite \eqref{E:costk} as
 \begin{align*}
 C_{k-1}(S_{k-1})=S_{k-1}\T \Phi_{k-1}S_{k-1}+\min_{\gamma^\theta_{k-1}}\{\sum_{i=1}^M\psi^i_{k-1}(S_{k-1})\theta^i_{k-1}\},
 \end{align*}
where $\Phi_{k-1}=H\T (\tilde N_k+\Phi_k)H$.

Let us denote
\begin{align*}
i^*_{k-1}(S_{k-1})=\arg\min_i\{\psi^i_{k-1}(S_{k-1})\}.
\end{align*}
If $i^*_{k-1}(S_{k-1})$ is not unique, then without loss of generality  one of the minimizers is randomly selected. Therefore,
\begin{align*}
\theta^*_{k-1}=\gamma^{\theta*}_{k-1}(S_{k-1})=b_{i^*_{k-1}(S_{k-1})},
\end{align*}
where $b_i\in \A$ is the $i$-th basis vector in $\R^M$. Thus,
\begin{align*}
C_{k-1}(S_{k-1})=S_{k-1}\T \Phi_{k-1}S_{k-1}+\min_i\{\psi^i_{k-1}(S_{k-1})\},
\end{align*}
and $C_k(S_k)$ is indeed of the form of \eqref{E:CK_assumptionMDP}. 

Let us now define the region
\begin{align*}
\tilde {\mathcal{R}}^i_{k-1}=\{s\in \R^{2n}~|~\psi^i_{k-1}(s)\le \psi^j_{k-1}(s) ~~\forall j\} \subseteq \R^{2n},
\end{align*}

and let us also define
\begin{align*}
\mathcal{R}^1_{k-1}&=\tilde{\mathcal{R}}^1_{k-1},\\
\mathcal{R}^i_{k-1}&=\tilde{\mathcal{R}}^i_{k-1}\setminus \cup_{j=1}^{i-1}(\tilde{\mathcal{R}}^i_{k-1}\cap\tilde{\mathcal{R}}^j_{k-1}).
\end{align*}
It can be verified that $\cup_{i=1}^M\tilde{\mathcal{R}}^i_{k-1}=\cup_{i=1}^M{\mathcal{R}}^i_{k-1}=\R^{2n}$. 
Moreover, $\mathcal{R}^i_{k-1}\cap \mathcal{R}^j_{k-1}=\varnothing$ for all $i,j\in \{1,2,\ldots,M\}$.
By construction, if $S_k\in \mathcal{R}^j_{k-1}$ for some $j$, then
$\psi^j_{k-1}(S_k)=\min_i\{\psi^i_{k-1}(S_k)\}$. 
Thus, $\gamma^{\theta*}_{k-1}(S_k)=b_j$ is an optimal strategy. In other words, $\theta^{i*}_{k-1}=1_{\mathcal{R}^i_{k-1}}(S_k)$ is optimal.
\end{proof}




\bibliographystyle{IEEEtran}
\bibliography{biblio}

\begin{IEEEbiography}[{\includegraphics[width=1in,height=1.25in,clip,keepaspectratio]{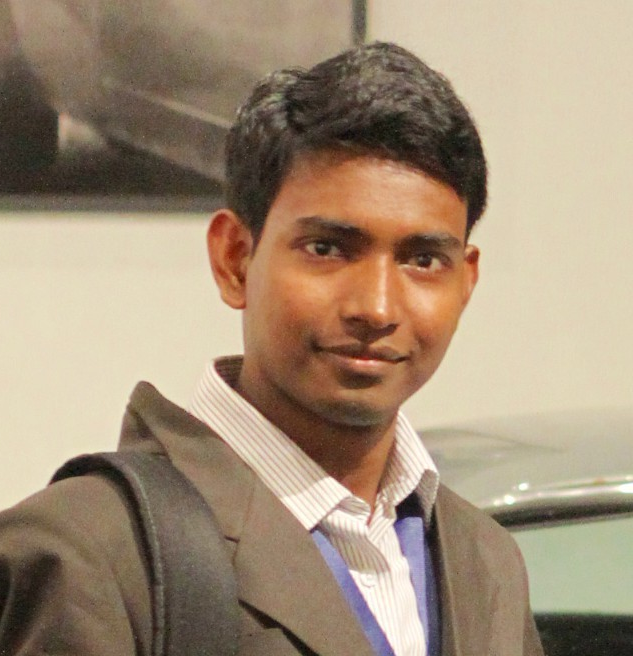}}]{Dipankar Maity}
 received the B.E. degree in Electronics and Telecommunication Engineering from Jadavpur University, India in 2013, and the Ph.D degree in Electrical and Computer Engineering from University of Maryland College Park, USA in 2018.  During his Ph.D, he was a visiting scholar at the Technische Universit\"{a}t M\"{u}nchen (TUM) and at the Royal Institute of Technology (KTH) Sweden. Currently, he is a Postdoctoral Fellow at Georgia Institute of Technology. 

His research interests include temporal logic based controller synthesis, Control with logical constraints, control with communication constraints, intermittent-feedback control,
event-triggered control, stochastic games, and integration of these ideas in the context of cyber-physical-systems. 
\end{IEEEbiography}

\begin{IEEEbiography}[{\includegraphics[width=1in,height=1.25in,clip,keepaspectratio]{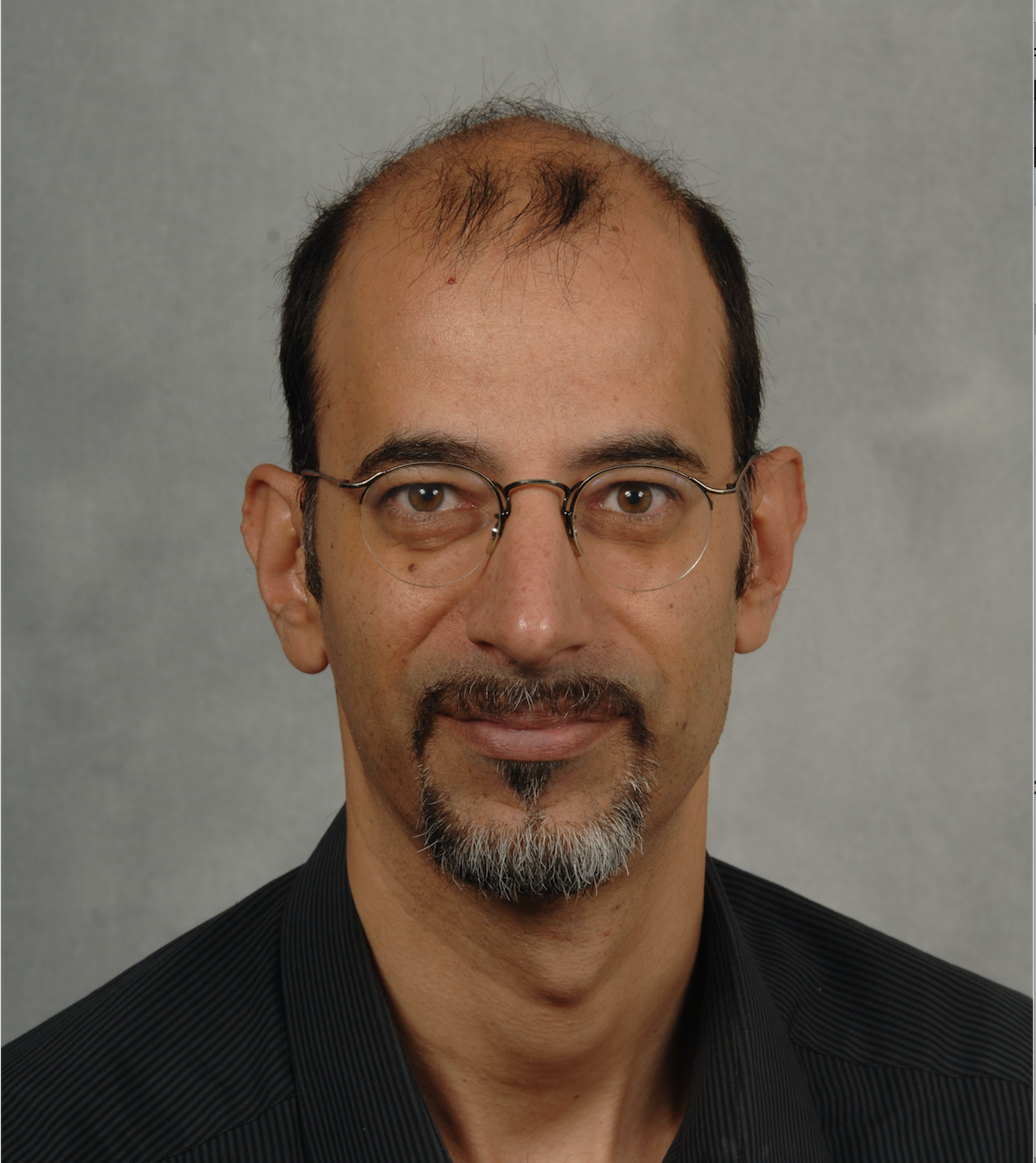}}]{Panagiotis Tsiotras}
  is the David and Andrew Lewis Chair Professor in the D. Guggenheim School of Aerospace Engineering at the Georgia Institute of Technology (Georgia Tech), and the Director of the Dynamics and Controls Systems Laboratory (DCSL) in the same school, as well as Associate Director of the Institute for Robotics and Intelligent Machines at Georgia Tech. He holds degrees in Aerospace Engineering, Mechanical Engineering, and Mathematics.

He has held visiting research appointments at MIT, JPL, INRIA Rocquencourt, and Mines ParisTech. His research interests include optimal control of nonlinear systems and ground, aerial and space vehicle autonomy. He has served in the Editorial Boards of the Transactions on Automatic Control, the IEEE Control Systems Magazine, the AIAA Journal of Guidance, Control and Dynamics, the Dynamic Games and Applications, and Dynamics and Control. He is the recipient of the NSF CAREER award, the Outstanding Aerospace Engineer award from Purdue, and the Technical Excellence Award in Aerospace Control from IEEE. He is a Fellow of AIAA, IEEE, and AAS, and a member of the Phi Kappa Phi, Tau Beta Pi, and Sigma Gamma Tau Honor Societies.

\end{IEEEbiography}

\end{document}